\documentclass[final,times]{elsarticle}

\usepackage{xcolor}
\usepackage{ amsmath}
\usepackage{amssymb, amsbsy}
\usepackage{graphicx}
\graphicspath{ {images/} }
\usepackage{epstopdf}
\usepackage{color, import}
\usepackage{sectsty}
\usepackage{enumitem}
\usepackage{lineno,hyperref}
\usepackage[linesnumbered,vlined,ruled]{algorithm2e}
\usepackage{setspace}
\usepackage[list=true]{subcaption}

\let\oldnl\nl
\newcommand{\nonl}{\renewcommand{\nl}{\let\nl\oldnl}}

\makeatletter
\def\ps@pprintTitle{%
 \let\@oddhead\@empty
 \let\@evenhead\@empty
 \def\@oddfoot{\centerline{\thepage}}%
 \let\@evenfoot\@oddfoot}
\makeatother

\sectionfont{\normalfont\fontsize{12}{16}\bfseries}
\subsectionfont{\normalfont\fontsize{10}{15}\bfseries}
\subsubsectionfont{\normalfont\fontsize{10}{13}\bfseries}

\newtheorem{lemma}{Lemma}
\newtheorem{theorem}{Theorem}

\newproof{proof}{Proof}
\newdefinition{remark}{Remark}
\newdefinition{example}{Example}

\begin{document}

\begin{frontmatter}

\title{Crash tolerant gathering on grid by asynchronous oblivious robots\tnoteref{mytitlenote}}

\tnotetext[fn1]{The first author would like to thank NBHM, DAE, Govt. of India and the second author would like to thank CSIR, Govt. of India for their financial support to carry out research. }

\author[math]{Kaustav Bose\corref{cor1}}
\ead{kaustavb@research.jdvu.ac.in}

\author[math]{Ranendu Adhikary}
\ead{ranendua@research.jdvu.ac.in}

\author[it]{Sruti Gan Chaudhuri}
\ead{srutiganc@it.jusl.ac.in}

\author[math]{Buddhadeb Sau}
\ead{bsau@math.jdvu.ac.in}

\address[math]{Department of Mathematics, Jadavpur University, West Bengal, India}

\address[it]{Department of Information Technology, Jadavpur University, West Bengal, India}

\cortext[cor1]{Corresponding author}


%
%

\begin{abstract}

Consider a system of autonomous mobile robots initially randomly deployed on the nodes of an anonymous finite grid. A gathering algorithm is a sequence of moves to be executed independently by each robot so that all robots meet at a single node after finite time. The robots operate in Look-Compute-Move cycles. In each cycle, a robot takes a snapshot of the current configuration of the grid in terms
of occupied nodes (\emph{Look}), then based on the perceived configuration, decides whether to stay put or to move to an adjacent node (\emph{Compute}), and in the later case makes an instantaneous move accordingly (\emph{Move}). The robots have \emph{weak multiplicity detection} capability, which enables them to detect if a node is empty or occupied by a single robot or by multiple robots. The robots are \emph{asynchronous}, \emph{oblivious}, \emph{anonymous}, can not communicate with each other and execute the same distributed algorithm. In a faulty system, however, any robot can \emph{crash}, which means that it becomes completely inactive and does not take part in the process any further. In that case a fault-tolerant gathering algorithm is an algorithm that gathers all the non-faulty robots at a single node. This paper considers a faulty system that can have at most one crash fault. With these assumptions  deterministic fault-tolerant gathering algorithms are presented that gather all initial configurations that are gatherable in a non-faulty system, except for one specific configuration called the \emph{2S2 configuration}.

\end{abstract}


\end{frontmatter}

\section{Introduction}\label{intro}

Robot swarms are a distributed system of autonomous mobile robots that collaboratively execute some complex tasks. Swarms of low-cost, weak, simple robots are emerging as a viable alternative to using a single powerful and expensive robot. Using swarm robot systems is particularly appealing while dealing with large-scale tasks in hostile or hazardous environments, as they can be more resilient to malfunctions or faults. Therefore devising algorithms for different complex tasks for a system of robots with  a minimal set of capabilities has recently received much attention from the distributed computing community. Considering the problems in fault-prone systems is particularly challenging. A comprehensive survey of different practical applications and research problems in multiple robot systems can be found in \cite{Cao97, Flocchini12}.

In swarm robot systems, \emph{gathering} is one of the most fundamental and widely studied problems. The problem is to devise a distributed algorithm that allows a system of weak robots, initially situated at different locations, to gather at some unspecified point within finite time and remain there. Many variants of the problem has been considered in literature with different assumptions on the robot capabilities as well as the underlying terrain on which they move. In this paper, we have considered a system of asynchronous and oblivious robots deployed on an anonymous finite grid. The robots are also prone to crash faults. The basic framework of our robotic system is described in detail in the following subsection.

\subsection{Basic model}

A system of $k$ ($k > 2$) robots is initially randomly deployed on the nodes of an \emph{anonymous} graph. By anonymous, it means neither the nodes nor the edges are labeled. In this paper, the input graph will be an $m \times n$ $(m \leq n)$ finite undirected grid embedded in the Euclidean plane.  The robots are assumed to be fully \emph{oblivious}, meaning that they have no memory of past configurations and previous actions. The robots are \emph{uniform} in the sense that they execute the same deterministic algorithm, and \emph{anonymous} in the sense that they are indistinguishable by their appearance and do not have any kind of identifiers. The robots are completely \emph{autonomous}, meaning there is no central control, no common coordinate system and no agreement on directions or chirality.  Furthermore, there are no means of communication between the robots. 

The robots have \emph{unlimited visibility} or \emph{global visibility}, that is they are able to sense the entire grid. The robots are assumed to be dimensionless, treated as points and do not obstruct the visibility or movement of any other robot. The robots are equipped with \emph{weak multiplicity detection} capability, which enables a robot to sense whether a node is occupied by a single robot or multiple robots. However, they can not ascertain the exact number of robots at a node occupied by multiple robots. We assume that initially there are no multiplicities, i.e., all the robots are initially located at distinct nodes. 

The robots, when active, operate according to the so-called \emph{LOOK-COMPUTE-MOVE} cycle. In each cycle, a previously idle or inactive robot wakes up and executes the following steps:

\begin{description}
 \item[LOOK:] The robot takes a snapshot of the current configuration of the grid. The configuration perceived by the robot is returned in form of an $m \times n$ matrix whose elements are from the set $\{0, 1, 2\}$, where 0 represents an empty node, 1 represents a singly occupied node and 2 represents a node occupied by multiple robots. We shall refer to a singly occupied node as a \emph{singleton} and a node occupied by multiple robots as \emph{multiplicity}. Since the grid is anonymous and there is no agreement on direction, a perceived configuration can be represented by many matrices. In particular, if $m = n$, there can be 8 different matrices satisfying the same configuration. 
 
 \item[COMPUTE:] Based on the perceived configuration, the robot performs computations according to a deterministic algorithm to decide whether to stay idle or to move to an adjacent node. As mentioned earlier, the deterministic algorithm is the same for all robots.
 
 \item[MOVE:] Based on the outcome of the algorithm the robot either remains stationary or makes an instantaneous move to an adjacent node. Since the moves are instantaneous, it implies that the robots are always seen on nodes, not on edges.
\end{description}

After executing a LOOK-COMPUTE-MOVE cycle, a robot becomes inactive. A robot may remain inactive or idle indefinitely, before waking up again to perform another LOOK-COMPUTE-MOVE cycle. However, unless the robot has crashed it does not remain inactive for infinite amount of time. We assume that the system of robots is \emph{fully asynchronous}, which means that the amount of time spent in LOOK, COMPUTE, MOVE and inactive states is finite (unless the robot has crashed) but unbounded, unpredictable and not same for different robots. As a result, the robots do not have
a common notion of time. Also the configuration perceived by a robot during the LOOK phase may significantly change before it actually makes a move. 

The robots are susceptible to \emph{crash-faults}. In a crash-prone system a robot can stop functioning at any time. It becomes completely inactive and does not take part in the process any further. However the robot is still physically present in the network, and can be perceived by all other non-faulty robots, but can not be identified as a crashed robot. In our system, we have assumed that at most one robot can crash. Our problem is to devise a distributed algorithm that gathers all the non-faulty robots at a single node.

\subsection{Related works}

The gathering problem has been extensively studied in continuous domain under various assumptions \cite{Cieliebak03, Cieliebak12,  Prencipe07, Flocchini05, Czyzowicz09, Lin07, Lin207}. In discrete domains the problem has been studied in different graph topologies.
The problem of gathering two robots on an anonymous ring was studied in \cite{Kranakis03, De06, Dessmark06}. The main difficulty of gathering problems is that the robots have to break symmetry to agree on a common meeting location. In \cite{Kranakis03} tokens were used to break symmetry, and in \cite{De06, Dessmark06} robots were assumed to have distinct identifiers. In \cite{Dessmark06} the robots move in synchronous steps, while in \cite{De06} the robots are asynchronous, but are allowed to meet inside an edge. The problem for $k \geq 2$ robots was considered in \cite{Flocchini04}, where the robots have memory and are allowed to use identical stationary tokens. The problem was first considered in a very minimal setting in \cite{Klasing08}, where the robots were assumed be identical, asynchronous, memoryless and without tokens or any kind of communication capability. They proved that without multiplicity detection gathering is impossible on rings for $k \geq 2$ robots. With weak multiplicity detection capability, they solved the problem for all configurations with an odd number of robots, and all the asymmetric configurations with an even number of robots by different algorithms. In \cite{Klasing10}, symmetric configurations with an even number of robots were studied, and the problem was solved when the number of robots is greater than 18. These left open the gatherable symmetric configurations with an even number of robots between 4 and 18, as the case of just 2 robots is ungatherable \cite{Klasing08}. Some of these configurations were solved in \cite{D11, Koren10, Haba08} in separate algorithms. In \cite{D14} a single unified algorithm was proposed that achieves gathering for all gatherable initial configurations except some potentially gatherable configurations with 4 robots. The problem was studied with weak local multiplicity detection in \cite{Izumi10, Kamei11, Kamei12}. A complete characterization of the problem on ring with local weak multiplicity detection has been provided in \cite{D14mini}. On finite grids, a full characterization of all gatherable configurations was given in \cite{Navarra16}. Furthermore, they showed that on these configurations gathering can be achieved even without any multiplicity detection capability. In \cite{Di17}, a full characterization of optimal gathering in infinite grid with global strong multiplicity detection was presented. The problem was investigated in \cite{Guilbault13} in asynchronous setting on regular bipartite graphs with weak multiplicity detection and limited visibility, i.e., the robots are only able to observe their neighboring nodes.

The gathering problem was first studied in fault-prone environment by Agmon and Peleg \cite{Agmon06}. They proposed an algorithm that solves the gathering problem in the plane in \emph{semi-synchronous} model in presence of at most one crash fault. Furthermore they proved that there exists no deterministic gathering algorithm in the semi-synchronous model that can tolerate a \emph{Byzantine robot}.  
A Byzantine robot is a faulty robot that behaves arbitrarily. They also showed that the gathering problem can be solved in the \emph{fully synchronous} model for a system of $n$ robots with $f$ Byzantine robots if $n \geq 3f + 1$. Fault-tolerant gathering algorithms in continuous domain have been studied for different types of faults under various assumptions \cite{Defago06, Bouzid13, Izumi12, bhagat16}. Fault-tolerant gathering on graphs were studied in \cite{Dieudonne14, Bouchard16, Chalopin16, Pelc17}. In \cite{Chalopin16} the agents were subject to \emph{delay faults}: the adversary delays the move of an agent for a few rounds. The problem in presence of Byzantine robots in synchronous setting was studied in \cite{Dieudonne14, Bouchard16}. The problem in presence of crash faults in asynchronous setting was first studied in the recent work by Andrzej Pelc in \cite{Pelc17}. Gathering algorithms using faulty tokens were studied in \cite{flocchini04token, das08}.  In \cite{das16}, a different type of fault was considered: faults were modeled as a malicious mobile agent that can block the path of the non-faulty agents and prevent them from gathering.


\subsection{Our contribution}

In this work we have considered the gathering problem for a system of asynchronous, oblivious, anonymous robots on an anonymous finite grid where there can be at most one crash fault. To the best of our knowledge the only work till date on gathering on graphs in presence of crash faults is the the recent work by Andrzej Pelc in \cite{Pelc17}. However the model is different from ours. In \cite{Pelc17}, the movements are not instantaneous, the robots have different labels, infinite memory, can exchange information with other robots when they meet on a node or an edge and can wait for a time of its choice at any node. In terms of the model, the work closest to ours is \cite{Navarra16} where it is shown that any non-partitive configuration of anonymous, asynchronous, oblivious robots on a finite grid is gatherable without any multiplicity detection. But in their algorithms there are many configurations where at most one robot is allowed to move at a time. Hence in those situations the algorithm can not survive a crash fault. The algorithm proposed in \cite{Di17} for infinite grids also works for finite grids. However it requires strong multiplicity detection capability. In this paper we have addressed the problem where at most one robot can suffer a crash fault. Assuming weak multiplicity detection capability, we have devised fault-tolerant deterministic gathering algorithms for all configurations that are gatherable in non-faulty systems, except the 2S2 configuration. 

\section{Some basic results on feasibility of gathering}

%
%

In this section we present some basic results on feasibility of gathering on general graphs. For a detailed exposition the readers are referred to \cite{Navarra17}. Assume that a set of robots is randomly deployed on the nodes of a simple undirected connected graph $G=(V,E)$, with vertex set $V$ and edge set $E$. The configuration of the robots on the graph can be represented by the pair $(G,f)$, where $f:V \longrightarrow \{0, 1, 2\}$ is a function defined as, 
\[
  f(v) =
  \begin{cases}
    0 & \text{if $v$ is an empty node} \\
    1 & \text{if $v$ is a singleton} \\
    2 & \text{if $v$ is a multiplicity} \\
   
  \end{cases}
 \]



  An \emph{automorphism} on a configuration $(G,f)$ is a bijection $\varphi : V \longrightarrow V$ such that for all $u, v \in V$, 1) $u, v$ are adjacent if and only if $\varphi(u), \varphi(v)$ are adjacent, 2) $f(v)=f(\varphi(v))$. The set of all automorphisms of $(G,f)$ forms a group called the \emph{automorphism group} of $(G,f)$, denoted by $Aut(G,f)$. If $|Aut(G,f)|=1$, we say that $(G,f)$ is \emph{asymmetric}, otherwise it is \emph{symmetric}.

For an automorphism $\varphi \in Aut(G,f)$, let $<\varphi> \subseteq Aut(G,f)$ be the cyclic subgroup generated by $\varphi$. Elements of this groups are $\{\varphi^0, \varphi^1, \varphi^2,\dots, \varphi^{p-1}\}$, where $\varphi^0$ is the identity, $\varphi^k=\underbrace{\varphi\circ\varphi\circ\cdots\circ\varphi}_\text{$k$ times}$ and $p$ is the order of $\varphi$. If $H$ is a subgroup of $Aut(G,f)$, the \emph{orbit of a vertex $v \in V$ under the action of $H$} is the set $H_v$= \{$\sigma(v) | \sigma\in H$\}.

\begin{description}
 \item [Partitive configuration:] Let $\mathcal{C} = ((V, E),f)$ be a configuration. An automorphism $\varphi\in Aut(\mathcal{C})$ is said to be \emph{partitive} on $V'\subseteq V$ if $|H_u|=p$ for all $u\in V'$, where $p>1$ is the order of $\varphi$ and $H=\{\varphi^0, \varphi^1, \varphi^2,\dots, \varphi^{p-1}\}$. $\mathcal{C}$ will be called a \emph{partitive configuration} if there is a $\varphi\in Aut(\mathcal{C})$ partitive on $V$.
 
\end{description}

%
%
%
%

\begin{theorem}[\cite{Navarra17}] Let $\mathcal{C}=((V, E), f)$ be a non-final configuration. If there exists a $\varphi\in Aut(\mathcal{C})$ partitive on $V$ then $\mathcal{C}$ is not gatherable.
 
 
\end{theorem}

In case of a configuration on a finite grid it is easy to show that it can have two types of symmetries (non-trivial automorphisms): 1) reflection, defined by an axis of reflection which acts as a mirror, and 2) rotation, defined by a center and an angle of rotation.  The center of a rotation can be a vertex, or the center of an edge, or the center of the area surrounded by four vertices, whereas the angle of rotation can be $\frac{\pi}{2}$ or $\pi$. Reflection axis can pass through vertices or through the middle of edges. It is easy to see that a configuration on a finite grid is partitive if and only if it either has a reflectional symmetry having its axis not passing through any vertex or a rotational symmetry having its center not lying on a vertex. As a result we have the following theorem.

\begin{theorem}[\cite{Navarra16}] If a configuration $\mathcal{C}$ on a finite grid has a reflectional symmetry having its axis not passing through any vertex or a rotational symmetry having its center not lying on a vertex, then $\mathcal{C}$ is ungatherable.

\end{theorem}


\section{Gathering algorithms}

A fault-tolerant gathering algorithm requires to gather all the non-faulty robots at a single node. Hence in presence of a single crash fault, the final configuration may have exactly two occupied nodes: one singleton having the crashed robot and the other a multiplicity having all the non-faulty robots. Since the robots have weak multiplicity detection capability, they can identify this configuration and report a successful execution of a gathering algorithm. However the same configuration will also arise if all the robots except one have gathered at some node and the remaining one has not crashed, but is yet to reach the gathering point. In this case the algorithm should not terminate at this point. Hence whenever a configuration with one singleton and one multiplicity is formed, the robots at the multiplicity will terminate, while the robot at the singleton, if not faulty, will move towards the multiplicity. However for simplicity, in the remaining of the paper we shall say that \emph{gathering is accomplished} when one of the following configurations is created: 1) exactly one occupied node 2) two occupied nodes with one singleton and one multiplicity.
%

In the remaining of the paper we shall always name our grid as $ABCD$ as shown in figure \ref{demo1}. We associate to each corner two strings or sequences with elements from $\{0,1,2\}$. The two sequences associated with $D$ will be denoted by $\lambda_{DA}$ and $\lambda_{DC}$. $\lambda_{DA}$  will be defined as the following. Scan the grid from $D$ along $DA$ to $A$ and sequentially all grid lines parallel to $DA$ in the same direction. For each node put a $0, 1$ or $2$ according to whether it is empty, singleton or a multiplicity. The string or sequence of length $mn$ thus obtained is the sequence $\lambda_{DA}$. For example, in figure \ref{demo1} the  sequence $\lambda_{DA}$ is $010102000000000100010220100000001011$. Similarly we can define the other seven sequences $\lambda_{AB}$, $\lambda_{AD}$, $\lambda_{BA}$, $\lambda_{BC}$, $\lambda_{CB}$, $\lambda_{CD}$, and $\lambda_{DC}$. Any two of these sequences can be compared using the lexicographical order or dictionary order. Similar types of strings or sequences are commonly used in gathering algorithms for different graph topologies like grids \cite{Navarra16}, rings \cite{D14}, trees \cite{Navarra16}.


 \begin{figure}
\centering

    \def\svgwidth{0.5\textwidth}
    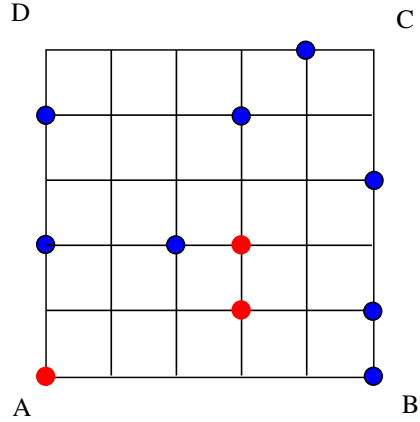

\caption{A configuration of robots on a grid. The blue nodes are singletons, and the red nodes are multiplicities.}
\label{demo1}
\end{figure}

We shall have different algorithms according to the structure of the grid, in particular, the length of the sides of the grid.
According to the length of the sides (the number of nodes on them), the grid can be classified into three types: odd $\times$ odd, even $\times$ odd and even $\times$ even. Algorithms for these cases are described in detail in the following subsections.

\subsection{Odd $\times$ odd grid}

In an odd $\times$ odd grid, the problem is trivially solvable. For any configuration of the robots, each of them can identify the unique central node of the grid regardless of their individual views. All the robots will be asked to move towards this node. This destination node will remain invariant under any movement of the robots, because it depends upon the very structure of the grid and not on the configuration of the robots. So all the non-faulty robots will eventually gather at the central node in finite time.

\subsection{Even $\times$ odd grid}

Without loss of generality assume that columns are of even length and the rows are odd. Consider the sequences $\lambda_{DA}$, $\lambda_{AD}$, $\lambda_{BC}$, and $\lambda_{CB}$ associated with the corners $D, A, B$ and $C$ respectively. If $\lambda_{DA} = \lambda_{AD}$ (which also implies $\lambda_{BC} = \lambda_{CB}$), the configuration is symmetric with the axis of symmetry passing through the middle of ${DA}$ and $CB$. Since ${DA}$ and $CB$ are even, the axis of symmetry is passing through edges. Thus the configuration is partitive and hence ungatherable. Similarly if $\lambda_{DA} = \lambda_{BC}$ (which implies $\lambda_{AD} = \lambda_{CB}$), then the configuration admits a $\pi$-rotational symmetry, and hence is partitive. If $\lambda_{DA} = \lambda_{CB}$ (implying $\lambda_{AD} = \lambda_{BC}$), the configuration is symmetric with the axis of symmetry passing through the middle of ${DC}$ and ${AB}$. Since ${DC}$ and ${AB}$ are odd, the axis of symmetry passes through nodes. Therefore this configuration is not partitive. So among the four sequences, at most two can be lexicographically largest. Furthermore, if there are two largest sequences, then they are associated to two corners of an odd side. Then based on the lexicographically largest sequence, we can choose between ${AB}$ and $DC$. Without loss of generality assume that ${DC}$ is associated with the largest sequence, i.e., $\lambda_{DA}$ or $\lambda_{CB}$ (or both) is the largest sequence. Now the grid can be divided into two equal halves by a line passing through the middle of the even sides of the grid. Call the half containing the side $DC$ the \emph{northern half}, and the one containing $AB$ the \emph{southern half}. We want to gather the robots at the central node of $DC$. The algorithm in this case is exactly the same as the one presented in \cite{Navarra16}, except for a minor modification to survive a potential crash fault. In \cite{Navarra16}, all the robots in the southern half are asked to move to the northern half until the southern half becomes completely empty. The robots in the northern half do not move until the southern half becomes completely empty. Once all the robots are in the northern half, they will all move to the center of the side $DC$. It can be proved that during the movements the northern half of the grid remains invariant, i.e., the largest sequence remains to be one of $\lambda_{DA}$ and $\lambda_{CB}$. However in the first phase of the algorithm, if a robot in the southern half crashes then the robots in the northern half will remain in wait mode eternally. This problem can be easily fixed by simply asking the robots in northern half to wait until all but one robot in the southern half reach the northern half. Thus we have the following theorem. 

\begin{theorem}
 All non-partitive configurations on even $\times$ odd grids are gatherable despite at most one crash fault.
\end{theorem}

\subsection{Even $\times$ even grid}\label{sec_even}

In the even $\times$ even case, we shall exclusively consider only square grids. The algorithm for a square grid can easily be modified for general rectangular grids. Rectangular grids with unequal sides are particularly easy to deal with as they do not admit any non-partitive symmetry and it is easier to fix a destination for gathering. In case of square grids we may have non-partitive symmetries when the axis of symmetry is a diagonal of the square. Breaking symmetries and devising algorithms in such cases are at times a complicated task. 

In case of an $n \times n$ even grid, the \emph{minimum enclosing square} or \emph{MES} is defined as the smallest square sub-grid, having the same geometric center as the original grid, that contains all the robots. The MES also is clearly even $\times$ even. Also the MES has at least one robot on the boundary. Based on whether the corners of the MES are occupied or not, we shall have five cases: no corners occupied, exactly one corner occupied, exactly two corners occupied, exactly three corners occupied and all four corners occupied. These cases will be discussed in detail in the following sections. We have devised our algorithms in such a way that the MES remains invariant while the robots move. Hence without loss of generality we shall assume that the original grid itself is the MES. This implies that there is at least one robot on the boundary of the grid.  

%
%
%
%
%
%
%
%
%
%

\subsubsection{Exactly one corner occupied}


\begin{theorem}\label{move0th}
 Any configuration on an even $\times$ even square grid with exactly one corner occupied is gatherable despite at most one crash fault.
\end{theorem}

\begin{proof}
 All the robots will be asked to move towards the occupied corner. A move by a robot is to be made in such a way that  1) its Manhattan distance from the occupied corner is reduced, and 2) in doing so it does not move to any other corner. Clearly the non-faulty nodes will gather at the initially occupied corner after finite time. $\square$
\end{proof}

%
%
%
%
%
%
%
%
%
%
%

\subsubsection{No corners occupied}

\begin{figure}[ht]
\centering
\subcaptionbox[Short Subcaption]{
       A purely asymmetric configuration \label{}
}
[
    0.3\textwidth 
]
{
    \def\svgwidth{0.3\textwidth}
    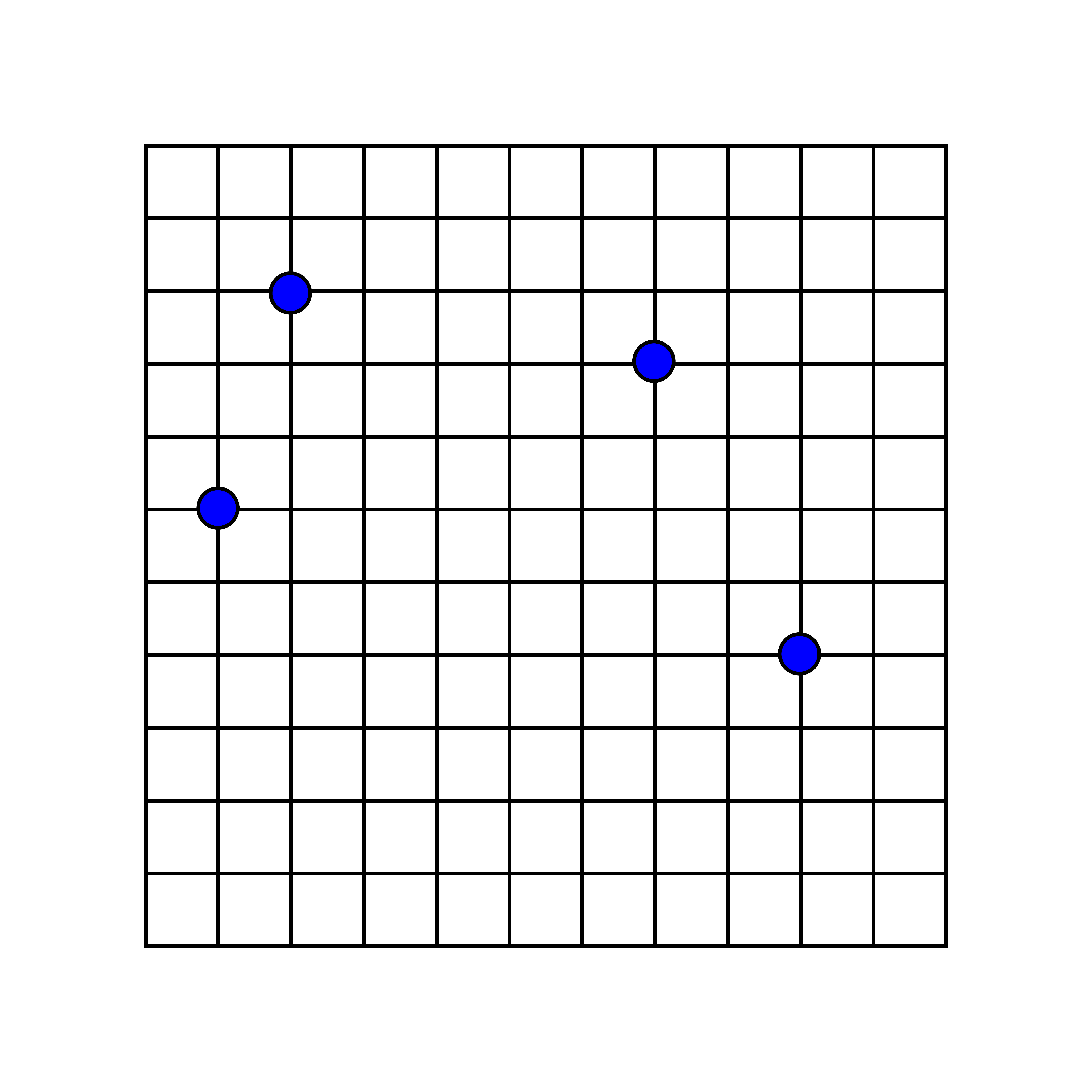
}
\hfill 
\subcaptionbox[Short Subcaption]{
    A symmetric configuration of the first type \label{}
}
[
    0.3\textwidth 
]
{
    \def\svgwidth{0.3\textwidth}
    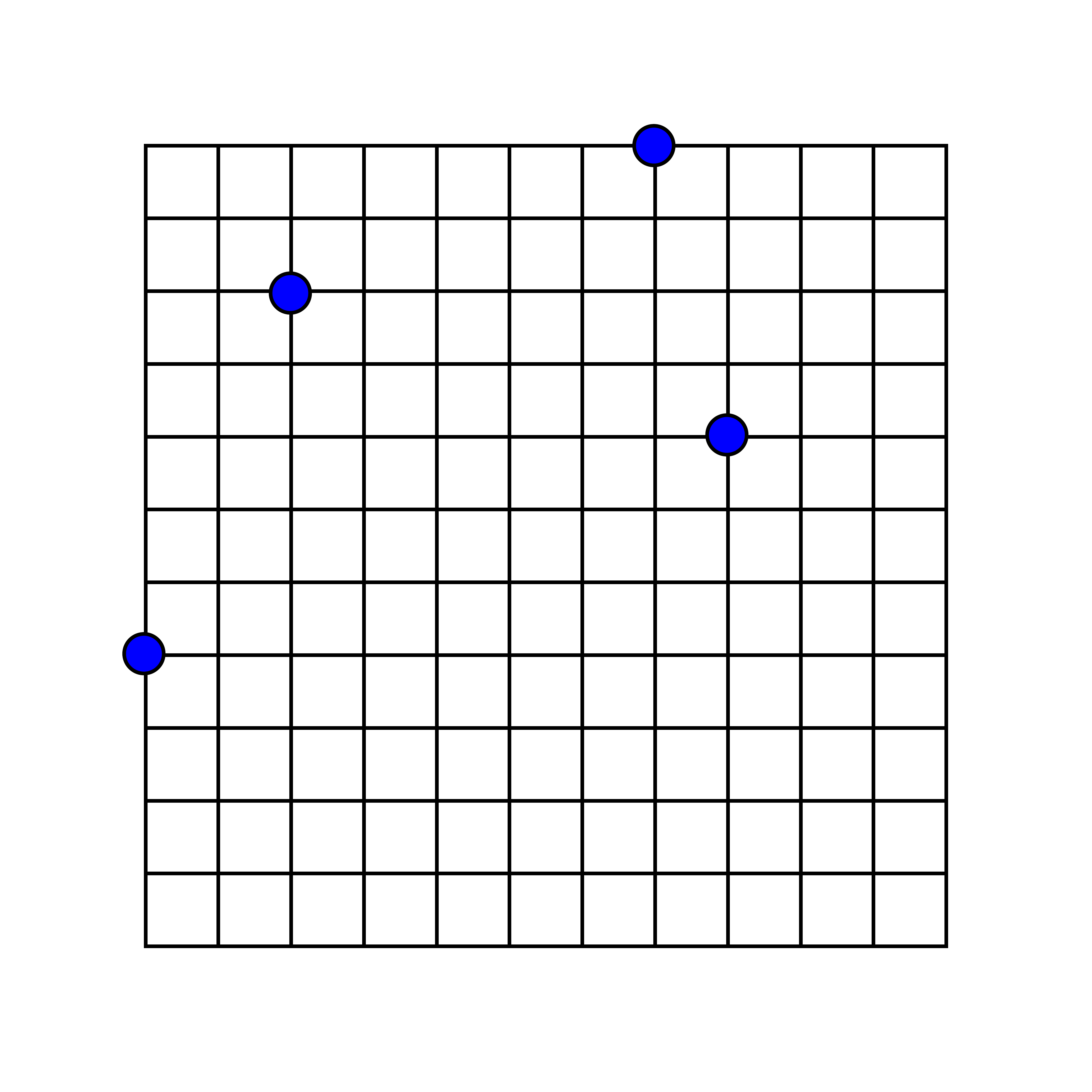
}
\hfill
\subcaptionbox[Short Subcaption]{
       An almost symmetric configuration of the first type \label{as1fig}
}
[
    0.3\textwidth 
]
{
    \def\svgwidth{0.3\textwidth}
    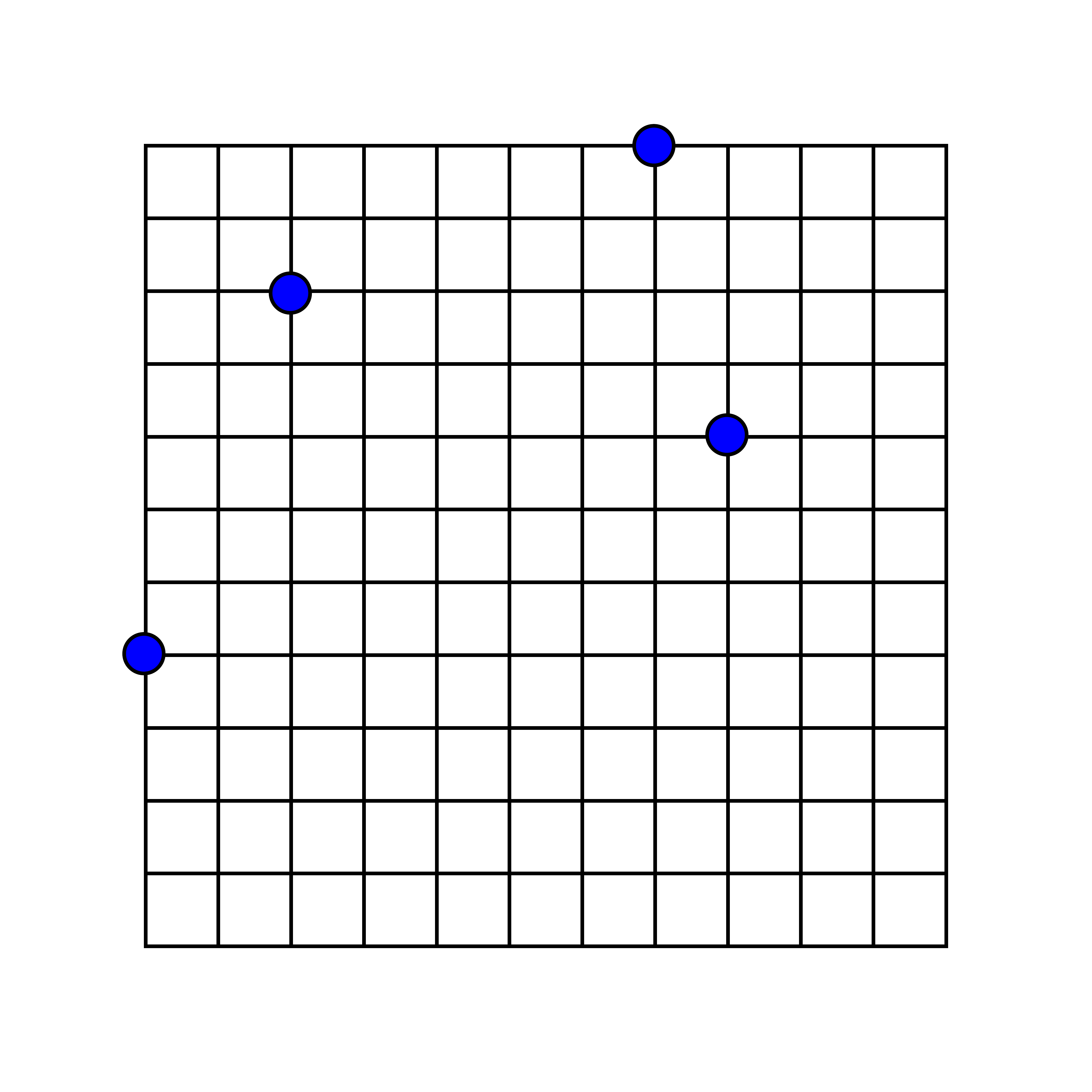
}
\\
\subcaptionbox[Short Subcaption]{
       A symmetric configuration of the second type \label{}
}
[
    0.3\textwidth 
]
{
    \def\svgwidth{0.3\textwidth}
    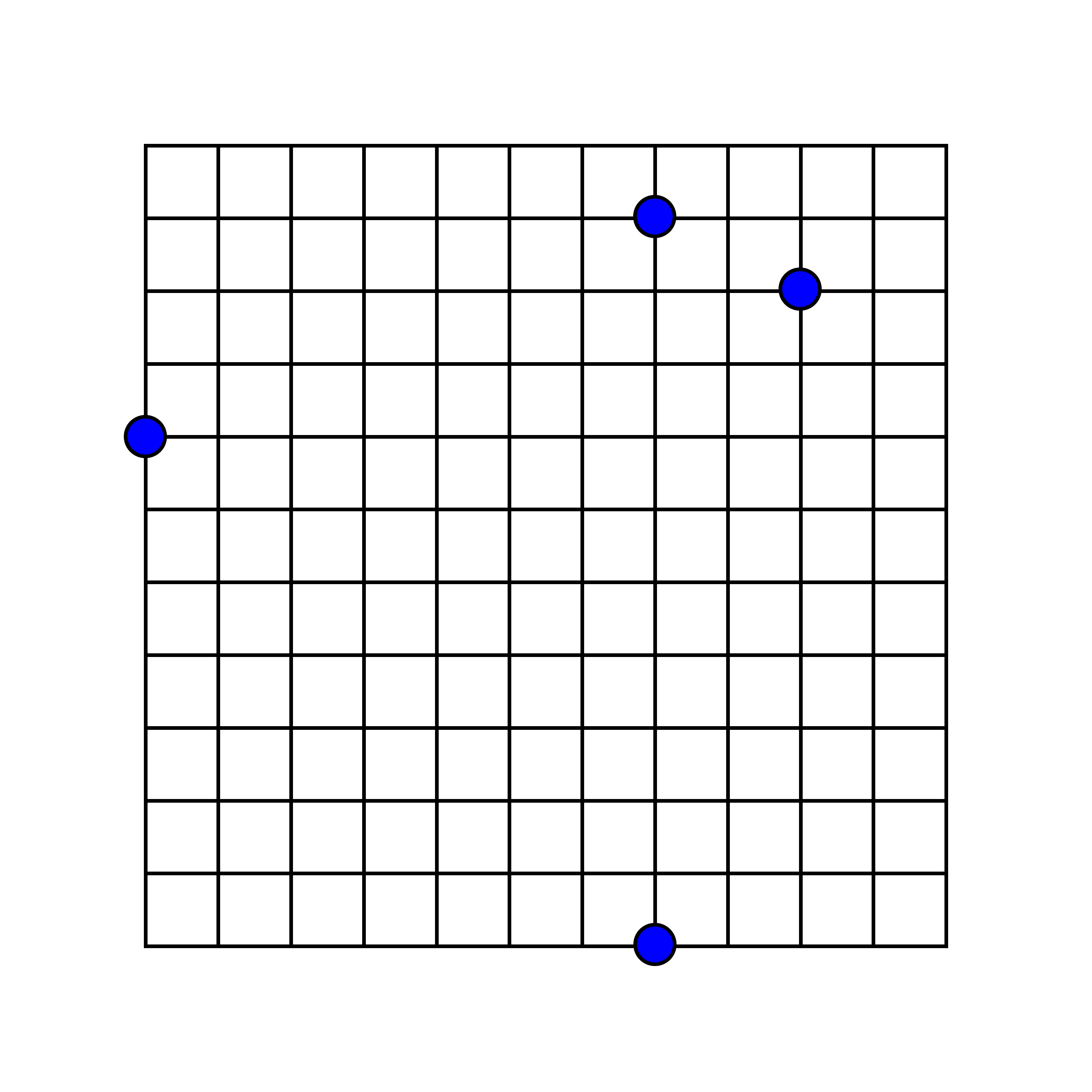
}
\hfill
\subcaptionbox[Short Subcaption]{
       An almost symmetric configuration of the second type \label{as2fig}
}
[
    0.3\textwidth 
]
{
    \def\svgwidth{0.3\textwidth}
    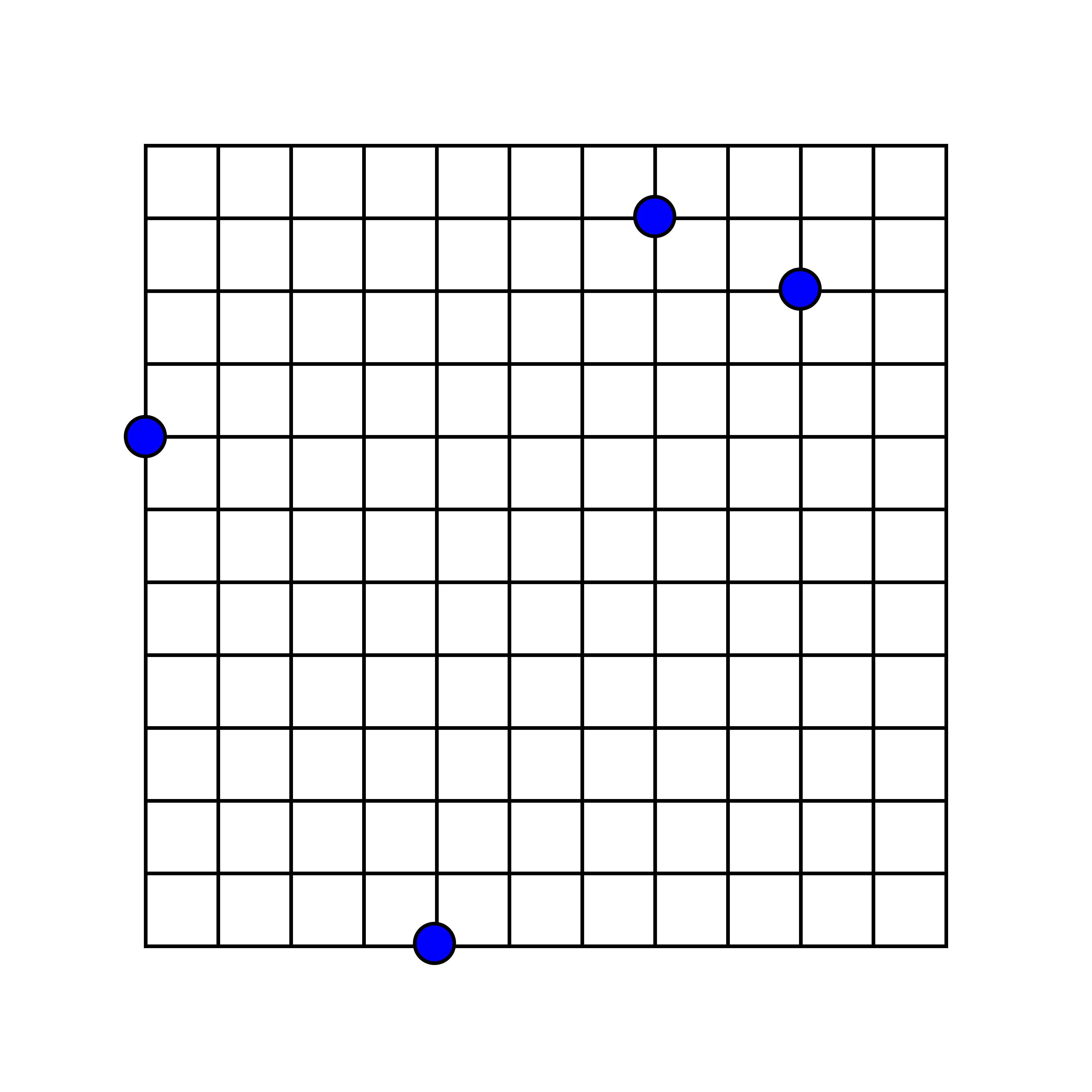
}
\hfill
\subcaptionbox[Short Subcaption]{
       A critical configuration \label{}
}
[
    0.3\textwidth 
]
{
    \def\svgwidth{0.3\textwidth}
    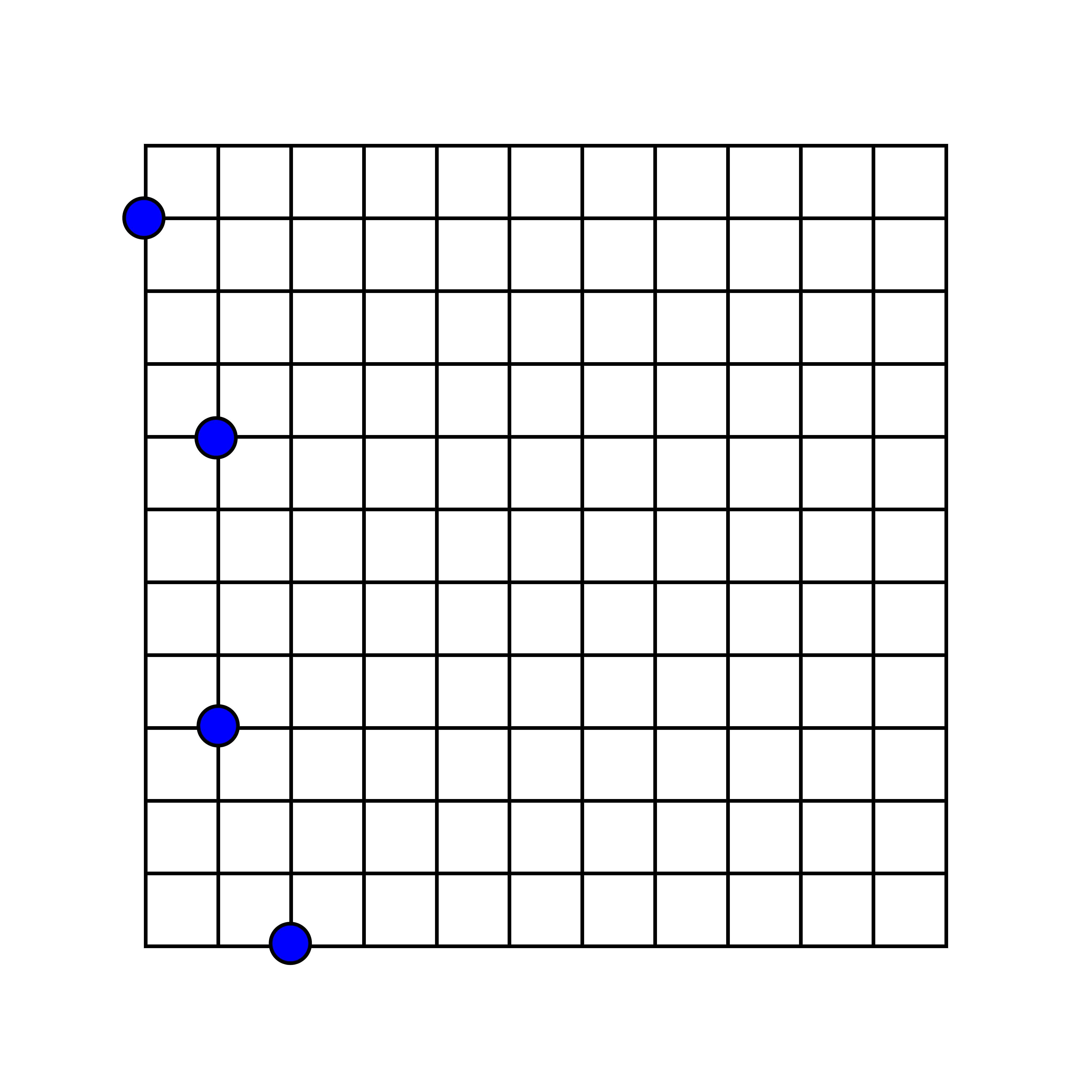
}

\caption[Short Caption]{Different types of configurations with no corners occupied. In each case $r_1$ and $r_2$ are the leading duo.}
\label{eg}
\end{figure}

Consider the case where none of the corners are occupied. However as mentioned earlier at least one robot is on the boundary of the grid. Now consider the sequences $\lambda_{AB}$, $\lambda_{AD}$, $\lambda_{BA}$, $\lambda_{BC}$, $\lambda_{CB}$, $\lambda_{CD}$, $\lambda_{DA}$, $\lambda_{DC}$. Suppose that $\lambda_{DA}$ is initially lexicographically largest. This implies that there must be at least one robot on $DA$. We shall call $D$ the \emph{largest corner}. Now we can have the following scenarios:

\begin{description}
 \item [Case 1]  $\lambda_{DA}$ is strictly the largest sequence. In that case the configuration is asymmetric. This is because if the configuration admits any symmetry, then $\lambda_{DA}$ must be equal to at least one other sequence.
 
 \item [Case 2] $\lambda_{DA} = \lambda_{DC}$. In this case the configuration is symmetric with the axis of symmetry  passing through $D$ and $B$.  This symmetry is not partitive. We call this symmetry as \emph{symmetry of the first type}.
 
 \item [Case 3] $\lambda_{DA} = \lambda_{AD}$. This implies that the configuration is symmetric with the axis of symmetry  passing through the middle of $DA$ and $CB$. Since the grid is even $\times$ even, the axis of symmetry does not pass though any node of the grid. Hence this is a partitive symmetry and thus ungatherable. 
 
 \item [Case 4] $\lambda_{DA} = \lambda_{CB}$. In this case the configuration is symmetric with the axis of symmetry  passing through the middle of $AB$ and $DC$, and hence partitive.
 
 \item [Case 5] $\lambda_{DA} = \lambda_{AB}$. In this case the configuration has a $\frac{\pi}{2}$-rotational symmetry. Similarly $\lambda_{DA} = \lambda_{CD}$  also leads to a $\frac{\pi}{2}$-rotational symmetry. This is a partitive symmetry, and hence the configuration is ungatherable.
 
 
 \item [Case 6] $\lambda_{DA} = \lambda_{BA}$ leads to a symmetric configuration with the axis of symmetry passing through $A$ and $C$. This symmetry is not partitive. We call this symmetry as \emph{symmetry of the second type}. 
 
 \item [Case 7] $\lambda_{DA} = \lambda_{BC}$. In this case the configuration has a $\pi$-rotational symmetry, which is a partitive symmetry.
\end{description}

 If the initial configuration has no corners occupied, then our strategy would be to occupy exactly one corner in the first phase of the algorithm. Then in the second phase the rest of the robots would gather at the occupied corner as described in theorem \ref{move0th}. Since there could be one crash fault, we need to ask at least two robots to move in the first phase. So accordingly we have to specify two robots, which we shall call the \emph{leading duo}, to move towards a particular corner of the grid. The leading duo have to be defined differently for the different types of configurations that we have classified, namely asymmetric, symmetric of the first type and symmetric of the second type. But before that we need to distinguish some particular type of asymmetric configurations, called \emph{almost symmetric configurations}. We classify the asymmetric configurations in the three following types:
 
 \begin{description}
  \item [Almost symmetric of the first type:]  Consider an asymmetric configuration with no corners occupied and no multiplicities. Since the configuration is asymmetric, there is a unique lexicographically largest sequence, say $\lambda_{DA}$. Hence there is at least one robot on $DA$. We shall call the configuration \emph{almost symmetric of the first type}, if there is at least one robot on the boundary edge $DC$ and the two sequences $\lambda_{DA}$ and $\lambda_{DC}$ become equal if the leading non-zero terms of both sequences are changed to 0. See figure \ref{as1fig}. In this case the robots corresponding to the leading non-zero terms of the two sequences associated with the largest corner, i.e. $\lambda_{DA}$ and $\lambda_{DC}$, are called the \emph{robots impeding symmetry}. 
  
  \item [Almost symmetric of the second type:] Consider an asymmetric configuration with no corners  occupied and no multiplicities. Since the configuration is asymmetric, there is a unique corner, say $D$, with which the strictly largest sequence is associated. Consider $A$ and $C$, i.e., the two corners other than the one with the largest sequence and its diagonally opposite corner. Now there are four sequences associated with $A$ and $C$. Again since the configuration is asymmetric, we can choose between $A$ and $C$ depending on the largest of these four sequences. Suppose it is $A$. We shall call the configuration \emph{almost symmetric of the second type} if  the two sequences associated with $A$, namely $\lambda_{AD}$ and $\lambda_{AB}$, become equal if the leading non-zero terms of both sequences are changed to 0. See figure \ref{as2fig} for an example. In this case $A$ will be called the \emph{second largest corner}, and the robots corresponding to the leading non-zero terms of the two sequences associated with the second largest corner are called the \emph{robots impeding symmetry}. Again note that the robots impeding symmetry must lie on the border of the grid.
  
  \item [Purely asymmetric:] The asymmetric configurations that are not almost symmetric will be called \emph{purely asymmetric}.
 \end{description}

 Now we are in the position to define the \emph{leading duo}. 
 
 \begin{enumerate}
 
  \item The leading duo of a purely asymmetric configuration are 1) the robots on the first two occupied places of the largest sequence when the first occupied node is a singleton, 2) the robots on the first occupied node of the largest sequence when the first occupied node is a multiplicity. In the first case the robots of the leading duo will be called the \emph{first robot} and the \emph{second robot}, according to the order in which they appear in the largest sequence.
  
  \item If the configuration is symmetric of the first type then the leading duo are the robots corresponding to the first non-zero terms of the two sequences associated with the largest corner.
  
  \item If the configuration is symmetric of the second type, then there are two diagonally opposite largest corners. Since the configuration is not partitive, we can choose between the two other corners by the largest sequences among the four sequences associated with them. This corner will be called the \emph{second largest corner}. The leading duo are the robots corresponding to the first non-zero terms of the two sequences associated with the second largest corner.  
  
  \item If the configuration is almost symmetric (of either types), then the leading duo would be the robots impeding symmetry.
 \end{enumerate}

 We shall discuss the purely asymmetric case later. The algorithm for the other configurations are described in lines 12-17 of algorithm \ref{Move0}. Notice that there could be an ambiguity if a symmetric configuration is both almost symmetric of the first type and almost symmetric of the second type. But we shall prove in lemma \ref{disjoint} that this can not happen.

 \begin{lemma}\label{disjoint}
  An asymmetric configuration can not be both almost symmetric of the first type and almost symmetric of the second type.
 \end{lemma}

 \begin{proof}
  If possible, consider an asymmetric configuration which is both almost symmetric of the first type and almost symmetric of the second type. Let $\lambda_{DA}$ be the (strictly) largest sequence. Let $r_1$ on $DA$ and $r_2$ on $DC$ be the robots impeding symmetry of the first type. Let $r_1$ and $r_2$ be at $x$th and $y$th node on $DA$ and $DC$ respectively from $D$ (in the sense that the 1st node from $D$ is $D$ itself). Clearly $x < y$. Now the configuration is also almost symmetric of the second type. Then either $A$ or $C$ is the second largest corner. But it implies from $x < y$ that $C$ must be the second largest corner.

    \textbf{Case 1:} Suppose that $r_2$ is not impeding symmetry of the second type. Then the robot on $BC$ closest to $B$ is at $y$th place from $B$. Call this robot $r_3$. Since the configuration is also almost symmetric of the first type, the robot on $AB$ closest to $B$, say $r_4$, is at $y$th place from $B$. This is contradiction to the fact that the configuration is almost symmetric of the second type. This is because the robot on $DA$ closest to $D$ is at $x$th place from $D$, while the robot on $AB$ closest to $B$ is at $y$th place from $B$, and $x \neq y$.
    
    \textbf{Case 2:} Now let $r_2$ be a robot impeding symmetry of the second type. In this case all the four boundary sides, namely $DA, AB, BC$ and $CD$, have exactly a single robot on them. Call them $r_1$, $r_4$, $r_3$ and $r_2$, respectively. Since the configuration is almost symmetric of the second type, $r_4$ is at $x$th place from $B$. As the configuration is also almost symmetric of the first type, $r_3$ is at $x$th place from $B$. Now lets compare the sequences $\lambda_{DA}$ and $\lambda_{BA}$.  The $(n(x-1)+1)$th term of $\lambda_{BA}$ is 1, corresponding to $r_3$. On the other hand the $(n(x-1)+1)$th term of $\lambda_{DA}$ is 0, because $r_2$ at $y$th position from  $D$ on $DC$ and $y > x$. Due to the almost symmetry of the second type, the first $n(x-1)$ terms of both the sequences are equal. Hence we have $\lambda_{DA} < \lambda_{BA}$, a contradiction. $\square$

 \end{proof}

 \begin{lemma}\label{almost_second_type_lemma}
 \begin{enumerate}
  \item Consider a non-partitive configuration which is symmetric of the second type or almost symmetric of the second type. Then one move made by (one or both of) the leading duo according to algorithm \ref{Move0} does not create a partitive configuration.
  
  \item Consider a non-partitive configuration which is symmetric of the first type or almost symmetric of the first type. Then one move made by (one or both of) the leading duo according to algorithm \ref{Move0} does not create a partitive configuration.
 \end{enumerate}

\end{lemma}
 
\begin{proof}
 \textbf{1.} We only prove for almost symmetric configurations of the second type. The proof for symmetric configurations of the second type is similar. Assume that  $\lambda_{DA}$ is the largest sequence and $A$ is the second largest corner. The proof would be exactly similar if $C$ were the second largest corner. Suppose that the robots on $AD$ and $AB$, that are closest to $A$ are respectively $r_1$ and $r_2$, and are at the $x$th and $y$th node from $A$ respectively. Since $\lambda_{DA}$ is the largest sequence and the configuration is almost symmetric of the second type, $x > y$.

 \textbf{Case 1 ($\boldsymbol{r_2}$ moves): } Suppose that $r_2$ makes a move towards $A$, and a partitive configuration is created. Now a partitive configuration can be created in three ways. The new configuration has either a horizontal symmetry (with respect to figure \ref{almost_second_type_fig_a}), or a vertical symmetry (figure \ref{almost_second_type_fig_b}), or a $\frac{\pi}{2}$-rotational symmetry (figure \ref{almost_second_type_leading_fig_c}), or a $\pi$-rotational symmetry (figure \ref{almost_second_type_leading_fig_d}).

   \textbf{Horizontal symmetry:} After the move on $AB$, the robot $r_2$ is now at $(y-1)$th place from $A$. Due to the horizontal symmetry there is a robot $r_2'$ on $DC$ at $(y-1)$th place from $D$. Since the configuration is almost symmetric of the second type, there is a robot $r_2''$ on $BC$ at $(y-1)$th place from $B$. Again due to the horizontal symmetry there is a robot $r_2'''$ on $BC$ at $(y-1)$th place from $C$. Since only the leading duo has moved, $r_2'''$ was at the same place initially. Recall that $r_2$ was closest to $A$ and was initially at $y$th place from $A$. But we see that $r_2'''$ was closer to $C$, than $r_2$ was to $A$. This contradicts the fact that $A$ is the second largest corner. 
   
   \textbf{Vertical symmetry:} The robot $r_2$ is now on $AB$ at $(y-1)$th place from $A$. Due to the vertical symmetry there is a robot $r_2'$ on $AB$ at $(y-1)$th place from $B$. Since the configuration is almost symmetric of the second type, there is a robot $r_2''$ on $DA$ at $(y-1)$th place from $D$. Again due to the vertical symmetry there is a robot $r_2'''$ on $CB$ at $(y-1)$th place from $C$. Similar to the last case this contradicts the fact that $A$ is the second largest corner.

    \textbf{$\boldsymbol{\frac{\pi}{2}}$-rotational symmetry:} Again the robot $r_2$ is on $AB$ at $(y-1)$th place from $A$. Due to the $\frac{\pi}{2}$-rotational symmetry there is a robot $r_2'$ on $BC$ at $(y-1)$th place from $B$. Again by the $\frac{\pi}{2}$-rotational symmetry there is a robot $r_2''$ on $DC$ at $(y-1)$th place from $C$. This contradicts the fact that $A$ is the second largest corner.

    \textbf{$\boldsymbol{\pi}$-rotational symmetry:} The robot $r_2$ is on $AB$ at $(y-1)$th place from $A$. Due to the $\pi$-rotational symmetry there is a robot $r_2'$ on $CD$ at $(y-1)$th place from $C$. Again this is a contradiction to the fact that $A$ is the second largest corner.

  \begin{figure}[thb!]
\centering
\subcaptionbox[Short Subcaption]{
       A horizontal symmetry is created after a move by $r_2$ \label{almost_second_type_fig_a}
}
[
    0.4\textwidth 
]
{
    \def\svgwidth{0.4\textwidth}
    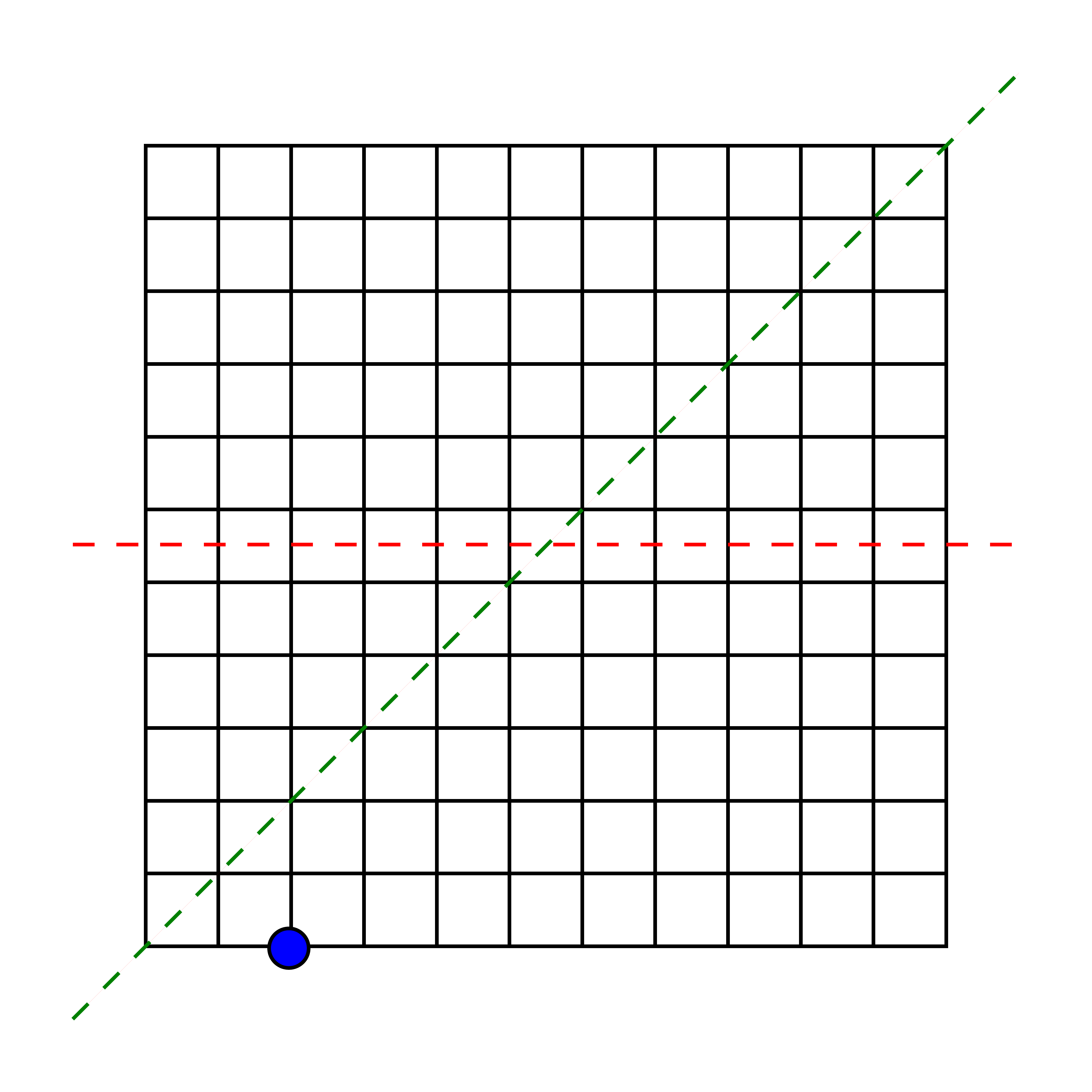
}
\hfill 
\subcaptionbox[Short Subcaption]{
    A vertical symmetry is created after a move by $r_2$ \label{almost_second_type_fig_b}
}
[
    0.4\textwidth 
]
{
    \def\svgwidth{0.4\textwidth}
    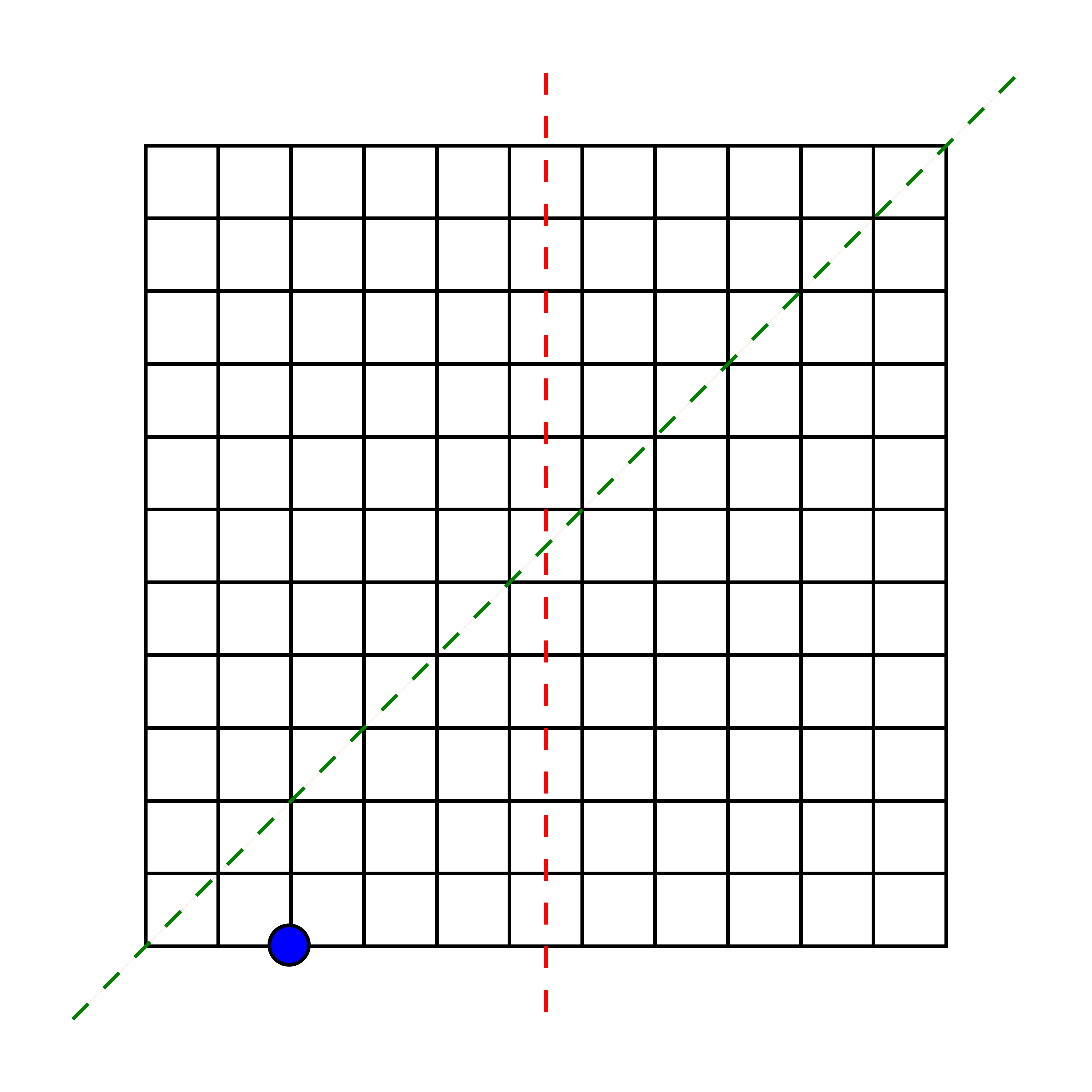
}
\\
\subcaptionbox[Short Subcaption]{
        A $\frac{\pi}{2}$-rotational symmetry is created after a move by $r_2$ \label{almost_second_type_leading_fig_c}
}
[
    0.4\textwidth 
]
{
    \def\svgwidth{0.4\textwidth}
    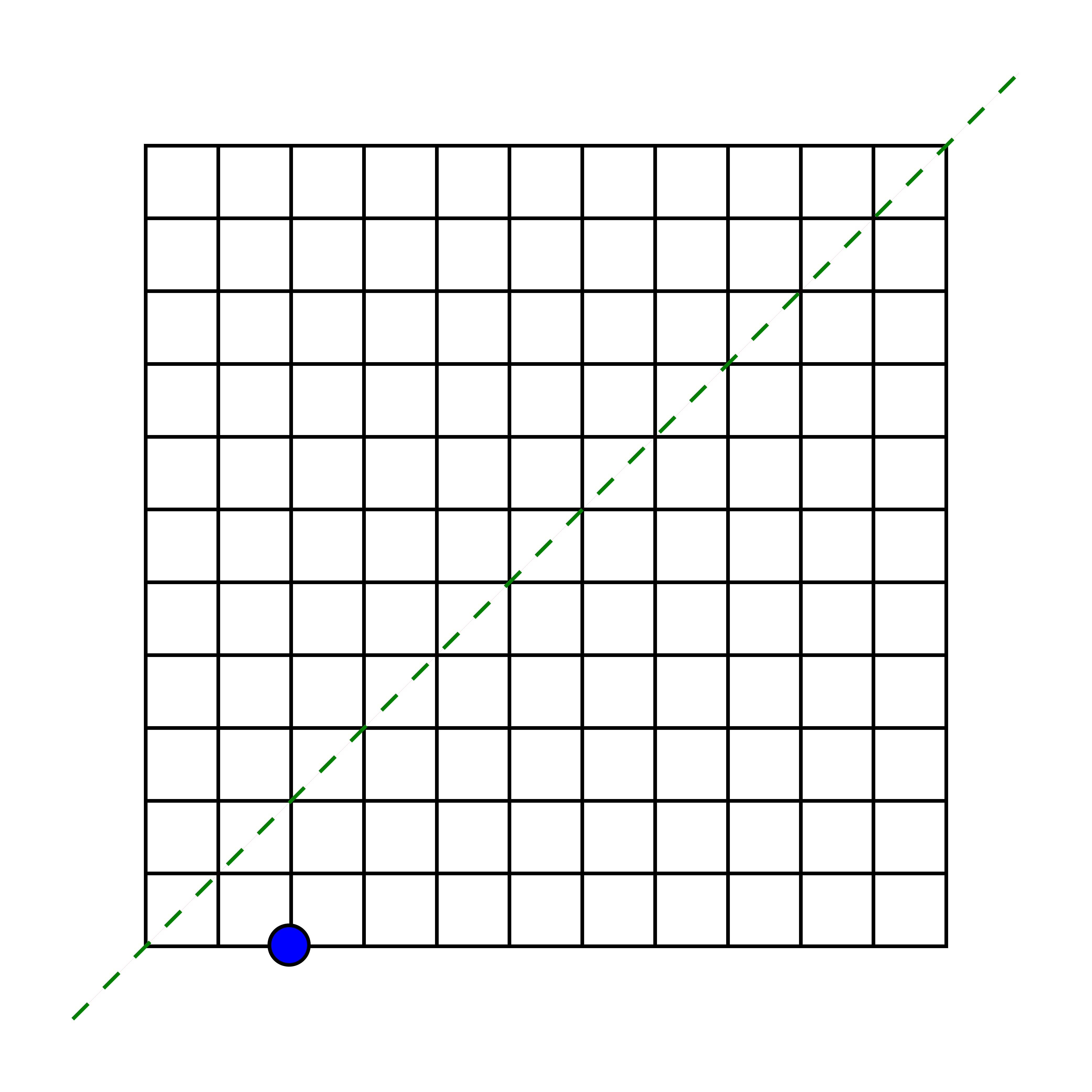
}
\hfill
\subcaptionbox[Short Subcaption]{
       A $\pi$-rotational symmetry is created after a move by $r_2$ \label{almost_second_type_leading_fig_d}
}
[
    0.4\textwidth 
]
{
    \def\svgwidth{0.4\textwidth}
    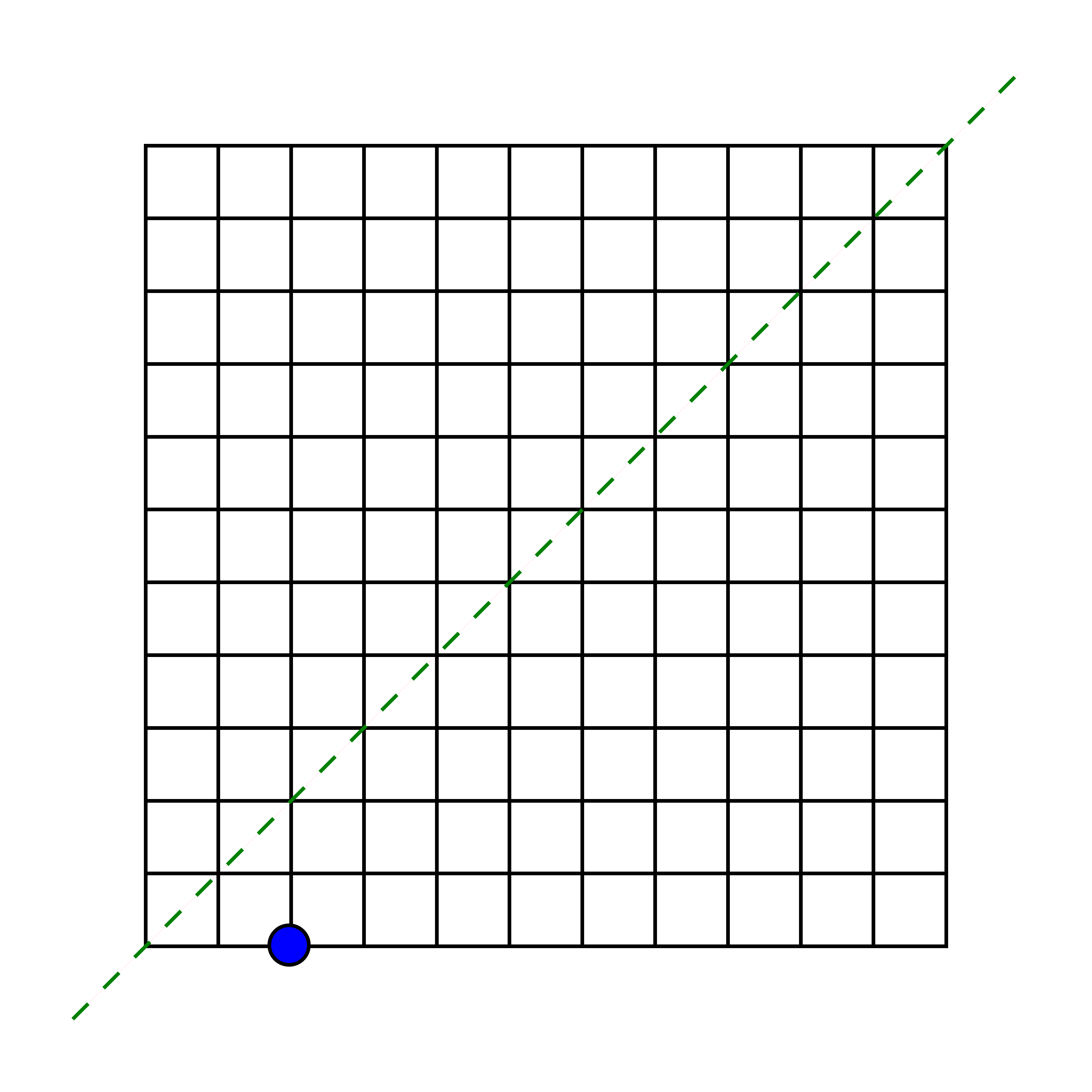
}

\caption[Short Caption]{Illustrations supporting the proof of case 1 of lemma \ref{almost_second_type_lemma}}
\label{almost_second_type_fig}
\end{figure}

  \textbf{Case 2 (Only $\boldsymbol{r_1}$ moves):} Now assume that only $r_1$ moves, and a partitive configuration is created.

   \textbf{Horizontal symmetry:} After the move on $DA,$ the robot $r_1$ is now at $(x-1)$th place from $A$. Since $r_1$ is the first robot on $AD$ from $A$, due to horizontal symmetry the first robot on $DA$ from $D$, say $r_1'$, is at $(x-1)$th place. Recall that $y<x$. But $y \nless x-1$, as $\lambda_{DA}$ was initially the largest sequence. Hence $y = x-1$, i.e., $r_2$ is at $(x-1)$th place from $A$ on $AB$. Due to the horizontal symmetry there is a robot $r_2'$ on $DC$, at $(x-1)$th place from $D$. Since the configuration is almost symmetric of the second type, there is a robot $r_2''$ on $BC$ at $(x-1)$th place from $B$. Again because of the horizontal symmetry, there is a robot $r_2'''$ on $BC$, at $(x-1)$th place from $C$. Since $r_1'$ is at $(x-1)$th place from $D$ and as the configuration is almost symmetric of the second type, there is a robot $r_1''$ on $AB$ at $(x-1)$th place from $B$. Again due to horizontal symmetry there is a robot  $r_1'''$ on $DC$, at $(x-1)$th place from $C$. Refer to the figure \ref{almost_second_type_trailing_fig_a}. Now let us compare the sequences associated with $A$ and $C$, before the move by $r_1$.
  
  \underline{$\boldsymbol{\lambda_{AD}^{old}}$ \textbf{vs} $\boldsymbol{\lambda_{CD}^{old}}$} 
  
  The first non-zero term of $\lambda_{AD}^{old}$ is the $x$th term, while the first non-zero term of $\lambda_{CD}^{old}$ is the $(x-1)$th term. Hence  $\lambda_{AD}^{old} < \lambda_{CD}^{old}$.
  
  \underline{$\boldsymbol{\lambda_{AB}^{old}}$ \textbf{vs} $\boldsymbol{\lambda_{CD}^{old}}$}
  
  Since only $r_1$ moves, the first $n(x-2)$ terms of the two sequences are unchanged. Clearly before the move the $n(x-2)+1$th term of $\lambda_{AB}$ was 0. On the other hand the $n(x-2)+1$th term of $\lambda_{CD}$ was 1, corresponding to the robot $r_2'''$. Hence $\lambda_{AB}^{old} < \lambda_{CD}^{old}$.
  
  This is a contradiction to the fact that $A$ is the second largest corner.
  
  \textbf{Vertical symmetry:} Due to the vertical symmetry there is a robot $r_2'$ on $AB$, at $y$th place from $B$, which is the closest robot to $B$ on $AB$. Since the configuration is almost symmetric of the second type, the first robot on $DA$ from $D$, say $r_2''$, is at the $y$th node from $D$. By the vertical symmetry we have $r_2'''$ on $CB,$ and then due to being almost symmetric of the second type we have  $r_2''''$ on $CD,$ both at the $y$th place from $C$ and closest to $C$. Again due to the vertical symmetry there is a robot $r_2'''''$ on $DC$ at the $y$th node from $D$. After the move $r_1$ is now at the $(x-1)$th node from $A$ on $AD$. Hence, by the vertical symmetry there is a robot $r_1'$ on $BC$ at the $(x-1)$th node from $B$. Now let us compare the sequences associated with $A$ and $C$, before the move by $r_1$.
  
  \underline{$\boldsymbol{\lambda_{AD}^{old}}$ \textbf{vs} $\boldsymbol{\lambda_{CB}^{old}}$} 
  
   Before the move, $r_1$ was at $x$th place from $A$ on $AD$. On $CB,$ the first robot from $C$ is $r_2'''$ at $y$th place from $C$. So clearly $\lambda_{AB}^{old} < \lambda_{CB}^{old}$, as $y < x$.
  
  \underline{$\boldsymbol{\lambda_{AB}^{old}}$ \textbf{vs} $\boldsymbol{\lambda_{CB}^{old}}$}

  We have $y<x \Rightarrow y \leq x-1$. First assume that $y < x-1$. For any of the eight sequences $\lambda$, we construct a sequence $\Lambda$ of $n$ terms by taking the first $n$ terms of $\lambda$ and replacing the last non-zero term with 0. Since only $r_1$ has moved to create the vertical symmetry, the first $n$ terms of $\lambda_{AB}^{old}$ and $\lambda_{BA}^{old}$ are equal. Hence also $\Lambda_{AB}^{old} = \Lambda_{BA}^{old}$.
  Since the configuration is almost symmetric of the second type, $\Lambda_{BA}^{old} = \Lambda_{DA}^{old}$. Also by the vertical symmetry, $\Lambda_{DA}^{old} = \Lambda_{CB}^{old}$. Hence $\Lambda_{AB}^{old} = \Lambda_{CB}^{old}$. This implies that in the first $n$ places of the sequences $\lambda_{CB}^{old}$ and $\lambda_{AB}^{old}$, all the non-zero terms, i.e. 1, are at identical places, except for the last non-zero term (corresponding to the robots $r_1'$ and $r_2'$ respectively). Now $r_1'$ is at the $(x-1)$th place from $B$ on $CB,$ while $r_2'$ is at the $y$th place from $B$ on $AB$. Since $y < x-1$, we clearly have $\lambda_{AB}^{old} < \lambda_{CB}^{old}$.
  
  Now let $y = x-1$.  This implies that the configuration has become symmetric of the second type after the move by $r_1$. Hence  $\lambda_{BA}^{new} = \lambda_{DA}^{new}$. By the vertical symmetry, $\lambda_{DA}^{new} = \lambda_{CB}^{new}$ and $\lambda_{AB}^{new} = \lambda_{BA}^{new}$. Hence we have $\lambda_{AB}^{new} = \lambda_{CB}^{new}$. Now since only $r_1$ has moved, the first $n(x-2)=n(y-1)$ terms of $\lambda_{AB}^{old}$ and $\lambda_{CB}^{old}$ have not changed and hence they are equal. The $(n(x-2) +1)$th term of $\lambda_{AB}^{old}$ is 0, since $r_1$ was initially at the $x$th place from $A$ on $DA$. On the other hand the $(n(x-2) +1)$th term of $\lambda_{CB}^{old}$ is 1, corresponding to the robot $r_2''''$. Hence $\lambda_{AB}^{old} < \lambda_{CB}^{old}$.

  This is a contradiction to the fact that $A$ is the second largest corner.

  \textbf{$\boldsymbol{\frac{\pi}{2}}$-rotational symmetry:} After the move on $DA,$ the robot $r_1$ is now at $(x-1)$th place from $A$, and is also the first robot on $DA$ from $A$. The robot $r_2$ is the first robot on $AB$ from $A$ at $y$th place from $A$. Due to the $\frac{\pi}{2}$ rotational symmetry there is a robot $r_1'$ on $AB$ at  $(x-1)$th place from $B$, and a robot $r_2'$ on $BC$ at $y$th place from $B$. Again there is a robot $r_1''$ on $BC$ at  $(x-1)$th place from $C$, and a robot $r_2''$ on $CD$ at $y$th place from $C$. Then since the configuration is almost symmetric of the second type $y = x-1$.

  \underline{$\boldsymbol{\lambda_{AD}^{old}}$ \textbf{vs} $\boldsymbol{\lambda_{CD}^{old}}$} 
  
  The first non-zero term of $\lambda_{AD}^{old}$ is the $x$th term corresponding to $r_1$, while the first non-zero term of $\lambda_{CD}^{old}$ is the $y=(x-1)$th term corresponding to $r_2''$. Hence  $\lambda_{AD}^{old} < \lambda_{CD}^{old}$.
  
  \underline{$\boldsymbol{\lambda_{AB}^{old}}$ \textbf{vs} $\boldsymbol{\lambda_{CD}^{old}}$}
  
  Due to the $\frac{\pi}{2}$ rotational symmetry $\lambda_{AB}^{new} = \lambda_{BC}^{new}=\lambda_{CD}^{new}$. Since only $r_1$ has moved, the first $n(x-2)=n(y-1)$ terms of $\lambda_{AB}^{old}$ and $\lambda_{CD}^{old}$ have not changed and hence they are equal. The $(n(x-2) +1)$th term of $\lambda_{AB}^{old}$ is 0, since $r_1$ was initially at the $x$th place from $A$ on $DA$. On the other hand the $(n(x-2) +1)$th term of $\lambda_{CD}^{old}$ is 1, corresponding to the robot $r_1''$. Hence $\lambda_{AB}^{old} < \lambda_{CD}^{old}$.
  
  This is a contradiction to the fact that $A$ is the second largest corner.
  
  \textbf{$\boldsymbol{\pi}$-rotational symmetry:} Due to the $\pi$ rotational symmetry there is a robot $r_1'$ on $BC$ at $(x-1)$th place from $C$, and a robot $r_2'$ on $CD$ at $y$th place from $C$. Then since the configuration is almost symmetric of the second type $y = x-1$. Exactly similar to the previous case we shall have $\lambda_{AD}^{old} < \lambda_{CD}^{old}$ and $\lambda_{AB}^{old} < \lambda_{CD}^{old}$. Again this contradicts the fact that $A$ is the second largest corner. $\square$

 \textbf{2.} The proof is similar to that of 1. 

\end{proof}

 \begin{figure}[thb!]
\centering
\subcaptionbox[Short Subcaption]{
       A horizontal symmetry is created after a move by $r_1$ \label{almost_second_type_trailing_fig_a}
}
[
    0.38\textwidth 
]
{
    \def\svgwidth{0.38\textwidth}
    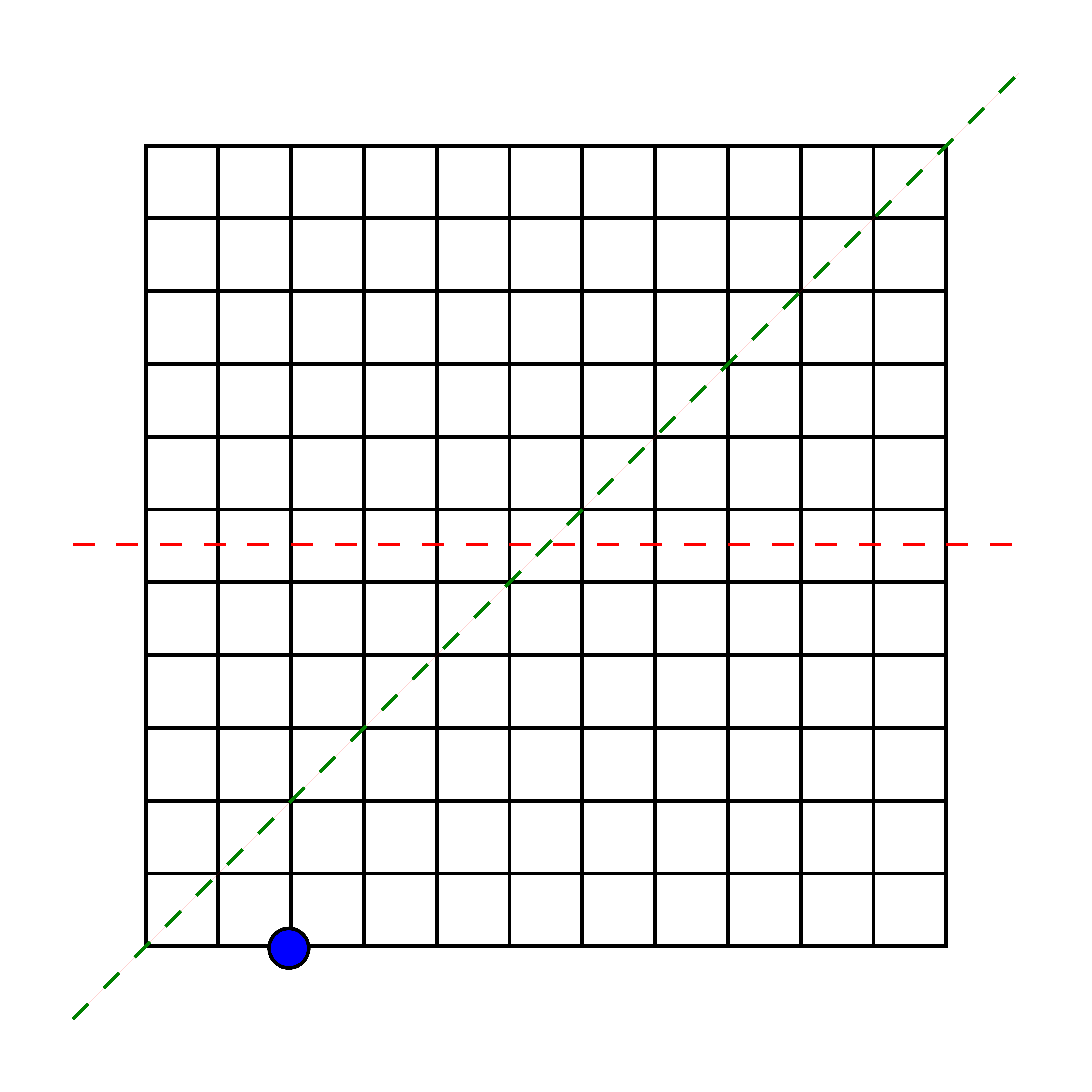
}
\hfill 
\subcaptionbox[Short Subcaption]{
     A vertical symmetry is created after a move by $r_1$ and $y < x-1$ \label{almost_second_type_trailing_fig_b}
}
[
    0.38\textwidth 
]
{
    \def\svgwidth{0.38\textwidth}
    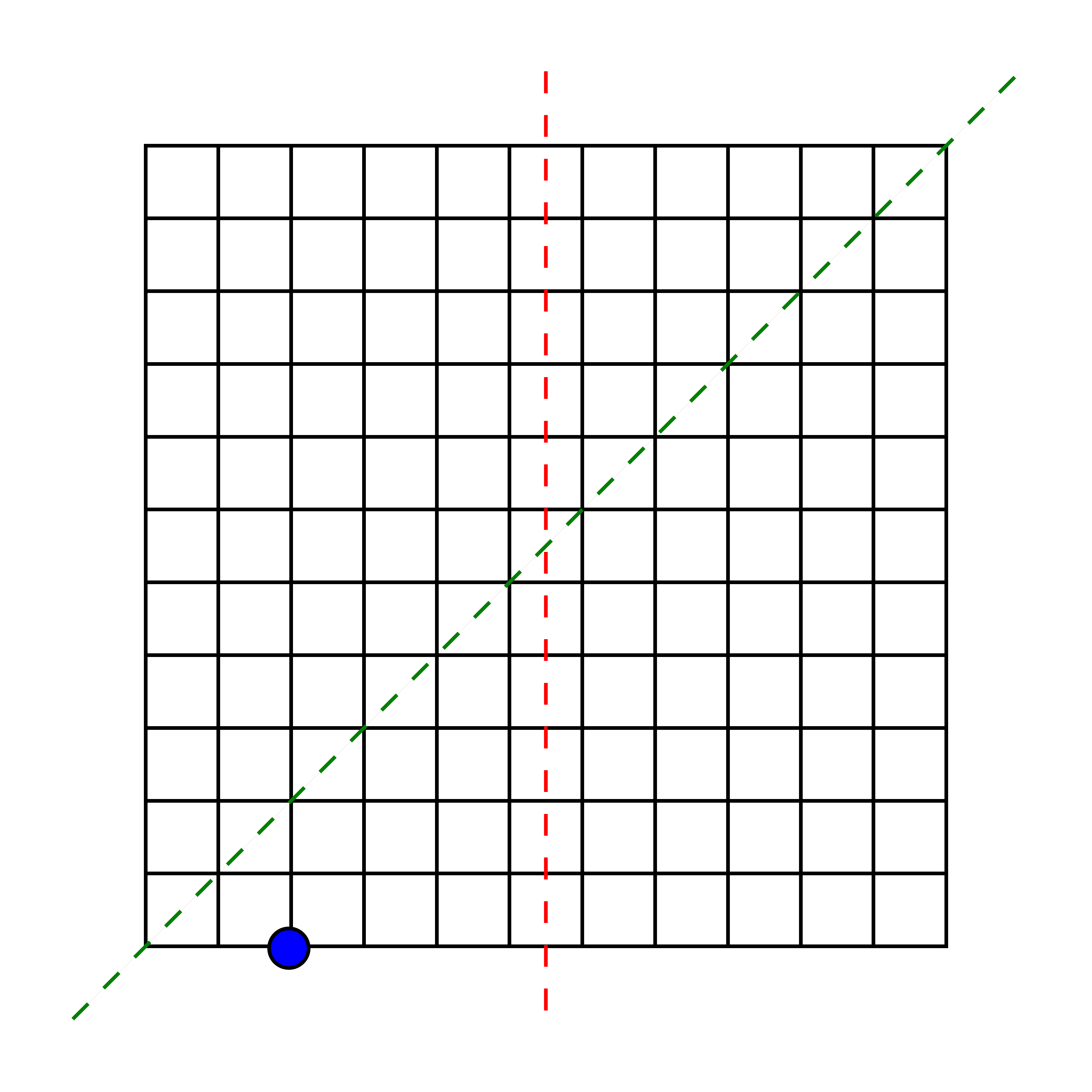
}
\\
\subcaptionbox[Short Subcaption]{
        A vertical symmetry is created after a move by $r_1$ and $y = x-1$ \label{almost_second_type_trailing_fig_c}
}
[
    0.38\textwidth 
]
{
    \def\svgwidth{0.38\textwidth}
    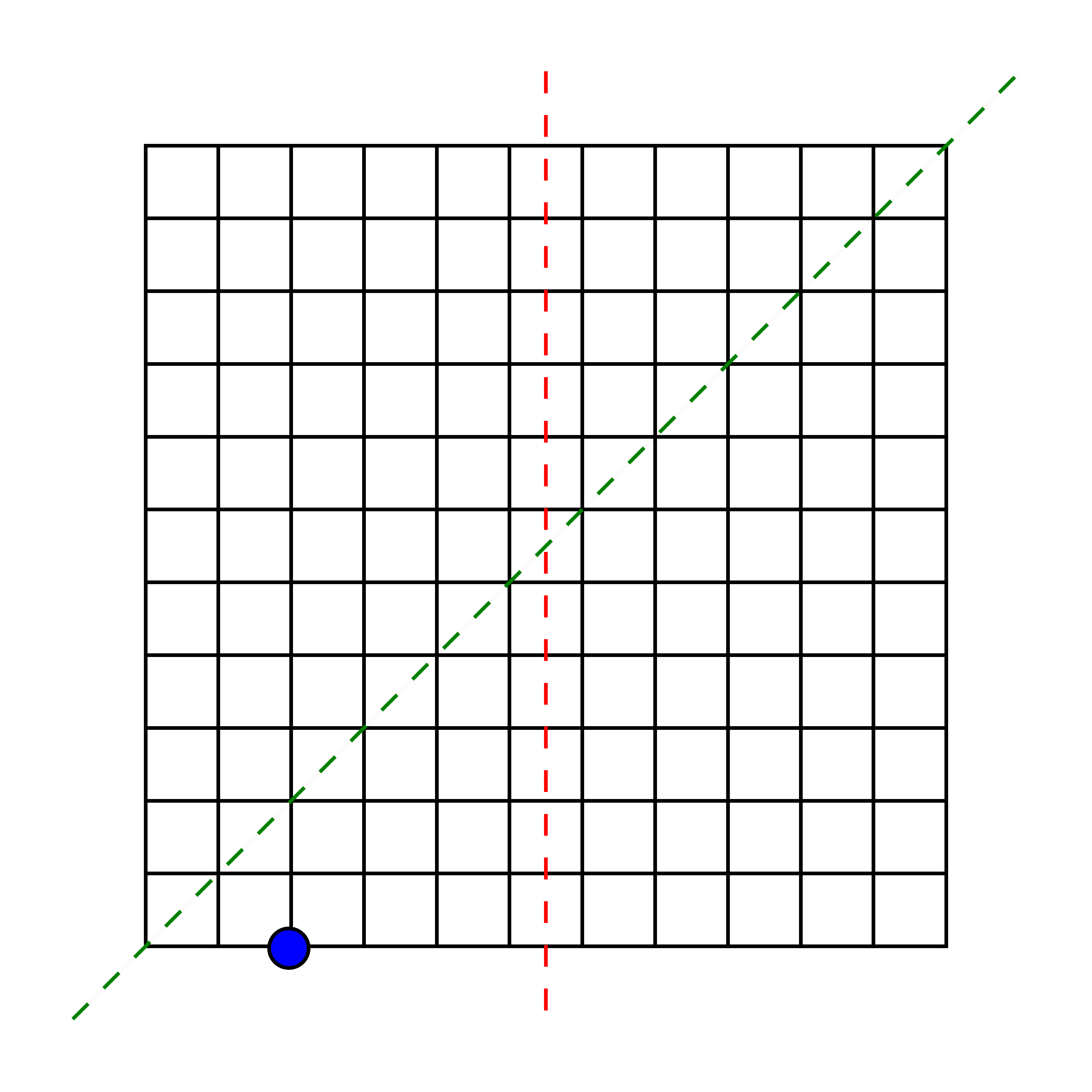
}
\hfill
\subcaptionbox[Short Subcaption]{
        A $\frac{\pi}{2}$-rotational symmetry is created after a move by $r_1$ \label{almost_second_type_trailing_fig_d}
}
[
    0.38\textwidth 
]
{
    \def\svgwidth{0.38\textwidth}
    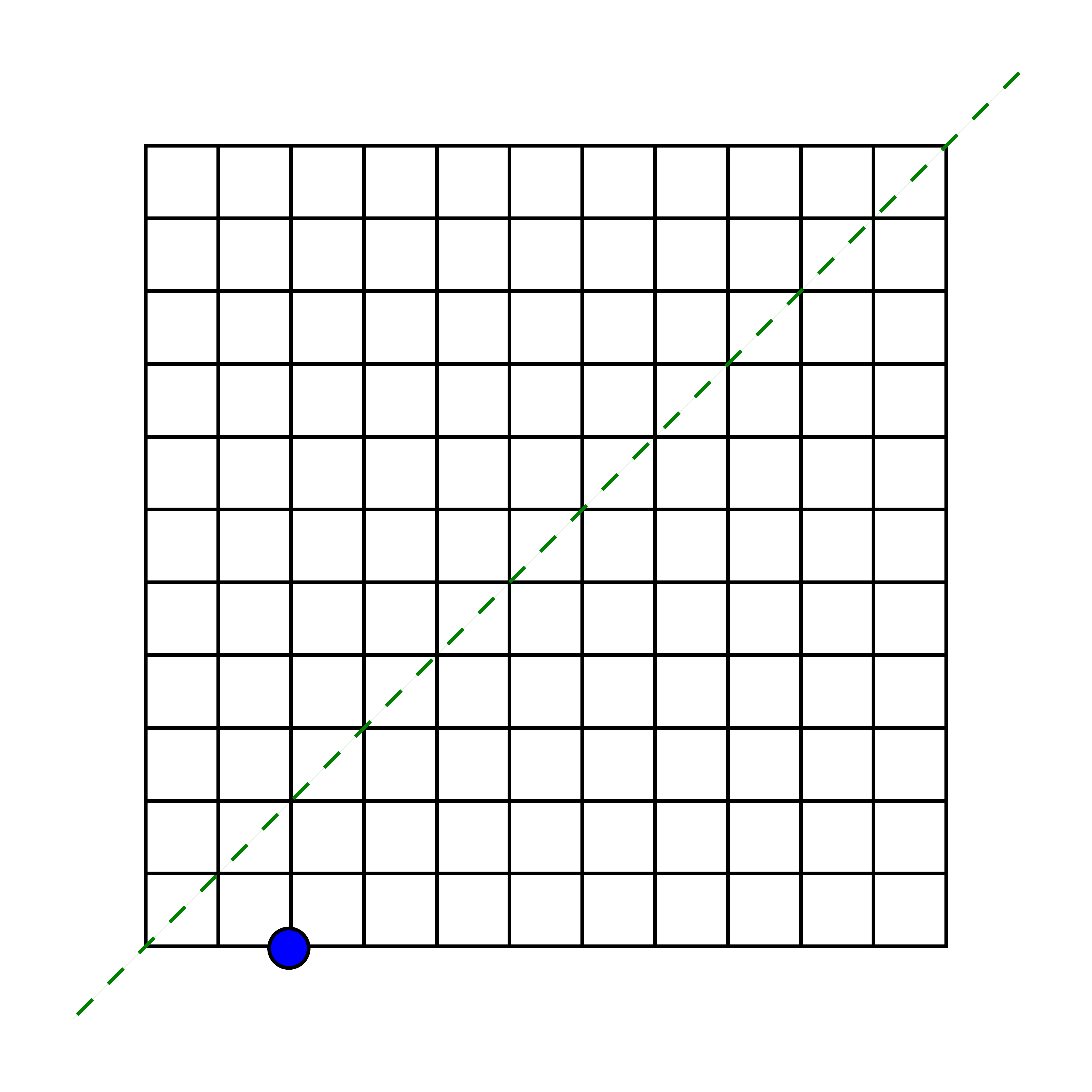
}

\subcaptionbox[Short Subcaption]{
        A $\pi$-rotational symmetry is created after a move by $r_1$ \label{almost_second_type_trailing_fig_e}
}
[
    0.38\textwidth 
]
{
    \def\svgwidth{0.38\textwidth}
    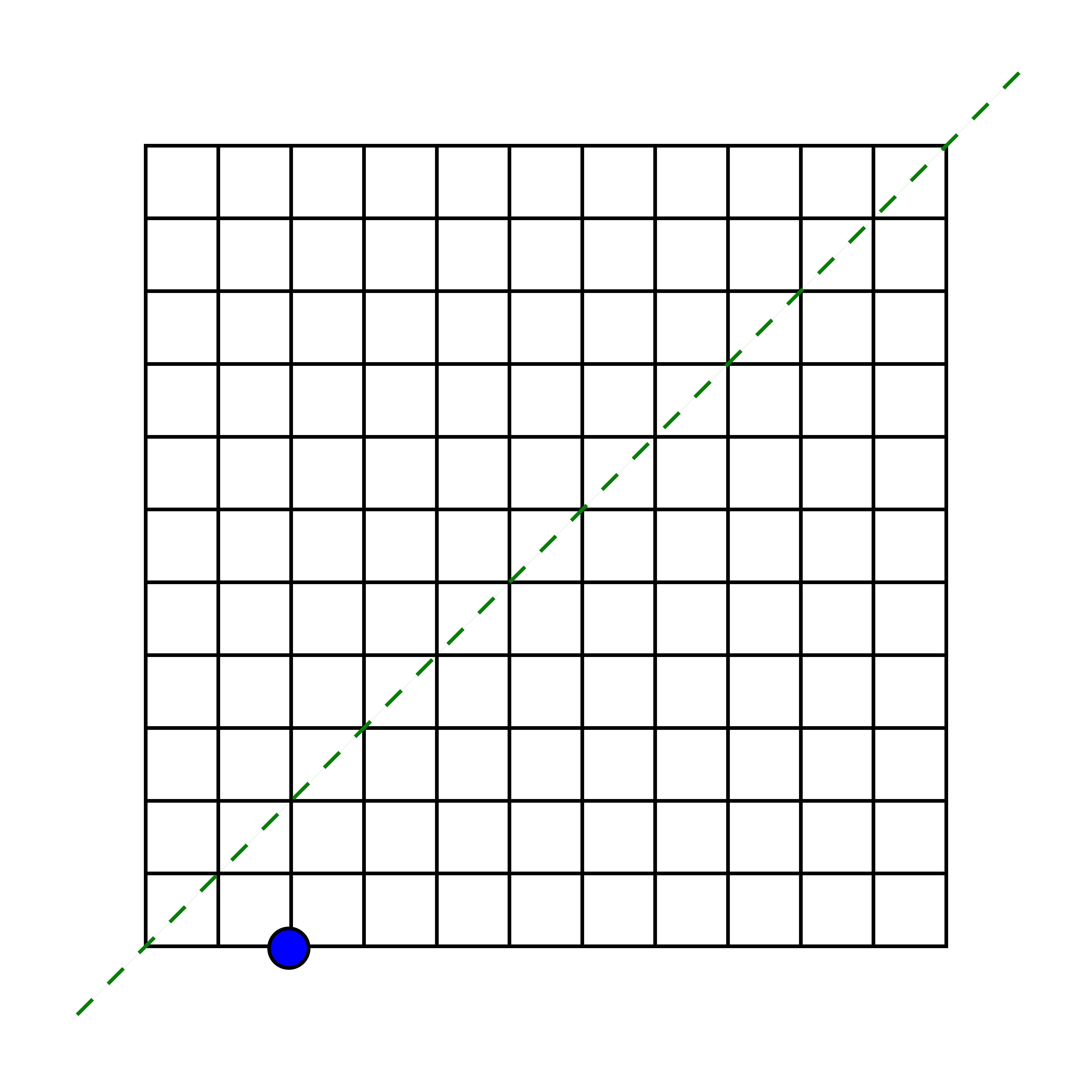
}

\caption[Short Caption]{Illustrations supporting the proof of case 2 of lemma \ref{almost_second_type_lemma}}
\label{almost_second_type_trailing_fig}
\end{figure}
 
  \begin{theorem}\label{second_type_theorem}
  \begin{enumerate}
      
   \item If the initial configuration is symmetric of the first type or almost symmetric of the first type, then algorithm \ref{Move0} leads to a configuration with exactly one corner occupied. 
   
   \item If the initial configuration is symmetric of the second type or almost symmetric of the second type, then algorithm \ref{Move0} leads to a configuration with exactly one corner occupied.
  \end{enumerate}

 \end{theorem}

 \begin{proof}
 
  \textbf{1.} Suppose that the initial configuration is symmetric of the first type. Also let $D$ be the largest corner, and $\{r_1, r_2\}$ be the leading duo. If they both move synchronously towards $D$, then the configuration remains symmetric of the first type with $D$ as the largest corner. If only one moves, say $r_1$, then the configuration becomes almost symmetric of the first type with $D$ as the largest corner. The leading duo in the new configuration is again $r_1$ and $r_2$, and should again move towards $D$. Note that $r_2$ may have a pending move towards $D$, and this is consistent with the algorithm. Similarly if the initial configuration is almost symmetric of the first type, then after any move by the leading duo the configuration becomes symmetric of the first type or again almost symmetric of the first type. In either case $D$ remains the largest corner, the leading duo is unchanged and any pending move is towards $D$. Hence if the initial configuration is symmetric of the first type or almost symmetric of the first type, the leading duo remain invariant and keep moving towards $D$ as the algorithm proceeds. Hence eventually $D$ gets occupied, even if one of the leading duo crashes.
  

  \textbf{2.} Now assume that the initial configuration is almost symmetric of the second type. Let $D$ be the largest corner, and $A$ be the second largest corner. After a move by one or both of the leading duo, the configuration is either asymmetric or symmetric. If the configuration is asymmetric and the largest corner is $A$, then it is almost symmetric of the first kind. If the largest corner is $D$ or $B$, then the configuration is almost symmetric of the second kind with $A$ being the second largest corner. Now suppose that the new configuration is symmetric. We have shown in lemma \ref{almost_second_type_lemma} that the symmetry is not partitive. Hence it must be a diagonal symmetry. The axis of symmetry can not be $DB$. This is because before the moves, $A$ was larger than $C$. So after moves towards $A$ and away from $C$, the corners $A$ and $C$ can not become symmetrical. Thus the axis of symmetry is $AC$. Then the configuration is either symmetric of the second type with $A$ as the second largest corner, or symmetric of the first type with $A$ as the largest corner. So we see that in each of the cases the leading duo remains the same and must continue to move towards $A$. Similarly if the initial configuration is symmetric of the second type with $A$ as the second largest corner, after any move the leading duo will remain unchanged and should again move towards $A$. Hence if the initial configuration is symmetric of the second type or almost symmetric of the second type, the leading duo remain invariant and keep moving towards the same corner as the algorithm proceeds. Any pending move arising due to the asynchronous environment is consistent with the algorithms for the constantly varying configurations. Therefore eventually exactly one corner gets occupied even if one of them crashes. $\square$
  
  \begin{algorithm}
    \setstretch{0.9}
    \SetKwInOut{Input}{Input}
    \SetKwInOut{Output}{Output}
    \SetKwProg{Fn}{Function}{}{}
    \SetKwProg{Pr}{Procedure}{}{}

    \Pr{\textsc{Move0()}}{

    \uIf{the configuration is purely asymmetric}{
    
    \uIf{the configuration is critical}{
    
    \If{$r$ is  one of the leading duo}{
    
    Move along the column towards the largest corner. 
    
    }

    }
    
    \Else{
    
    \If{ $r$ is one of the leading duo at $(i, j)$}{
    
    \uIf{($j = 1$) or ($j = 2$ and $i > 2$) or ($i = n$ and $j \leq \frac{n}{2}$)}{
    
    Move along the column towards the largest corner. \;
    
    }
    
    \Else{
    
    Move along the row towards the largest corner. \;
    
    }
    
    }

    }
    
    }
    
    \uElseIf{the configuration is symmetric of the first type or almost symmetric of the first type}{
    
    \If{$r$ is one of the leading duo}{
    
    Move towards largest corner \;
    
    }

    }
    
    \ElseIf{the configuration is symmetric of the second type or almost symmetric of the second type}{
    
    \If{$r$ is one of the leading duo}{
    
    Move towards second largest corner \;
    
    }

    }

    }
      
    \caption{\textsc{Move0}}
    \label{Move0} 
\end{algorithm}

 \end{proof}
 
  \begin{figure}[thb!]
\centering
\subcaptionbox[Short Subcaption]{
     $D$ is the largest corner and $\{r_1, r_2\}$ are the leading duo.   \label{counter_fig_a}
}
[
    0.4\textwidth 
]
{
    \def\svgwidth{0.4\textwidth}
    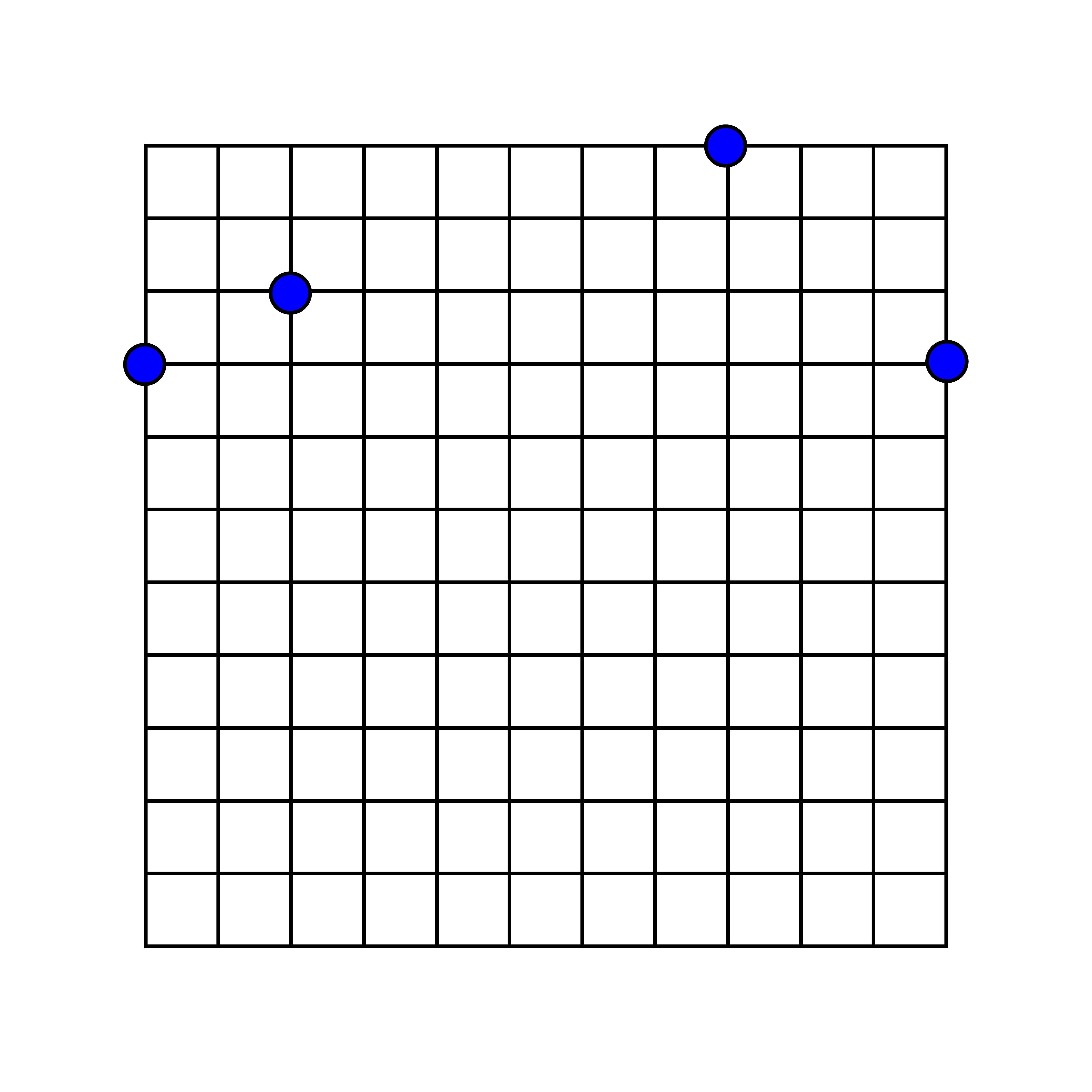
}
\hfill 
\subcaptionbox[Short Subcaption]{
    After a move by $r_2$, $C$ becomes the largest corner with $\{r_3, r_2\}$ as the leading duo.  \label{counter_fig_b}
}
[
    0.4\textwidth 
]
{
    \def\svgwidth{0.4\textwidth}
    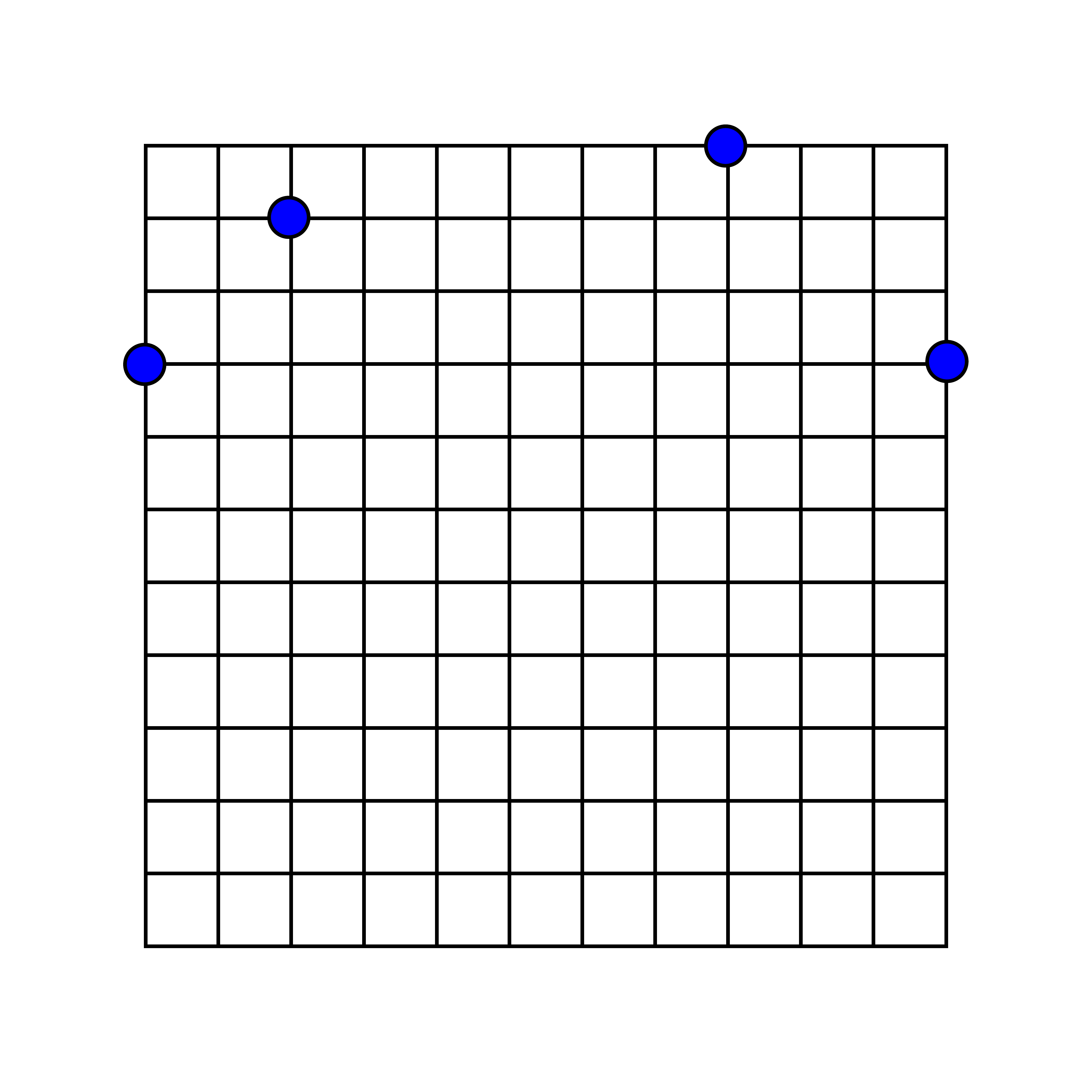
}
\\
\subcaptionbox[Short Subcaption]{
       After simultaneous moves by $r_2$ and $r_3$, the configuration becomes partitive. \label{counter_fig_c}
}
[
    0.4\textwidth 
]
{
    \def\svgwidth{0.4\textwidth}
    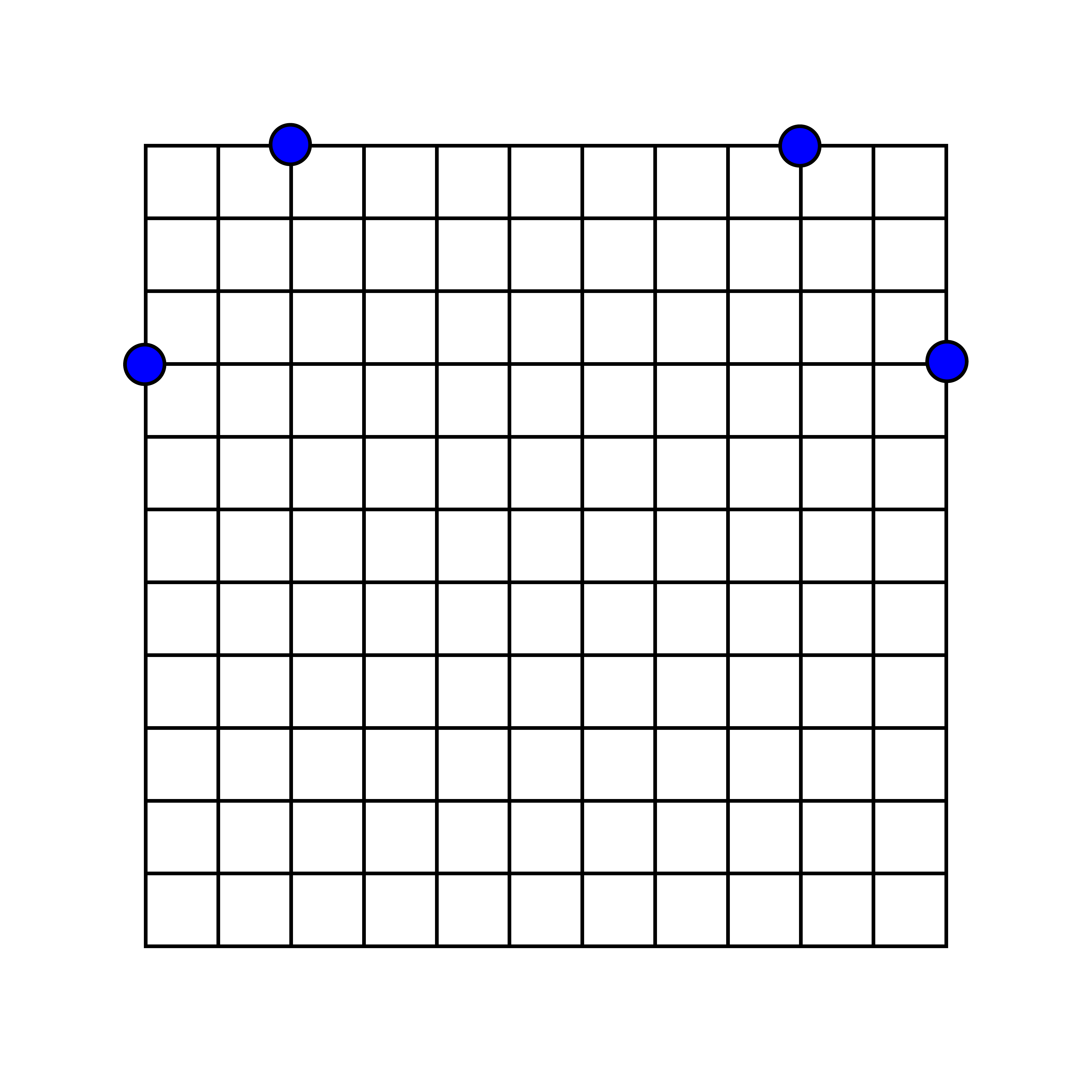
}

\caption[Short Caption]{}
\label{counter_fig}
\end{figure}

 Now we discuss the purely asymmetric case. Suppose that $\lambda_{DA}$ is the strictly largest sequence. Then the most obvious algorithm for the purely asymmetric case, would be to ask the `first two' robots (the leading duo in this case) in the sequence $\lambda_{DA}$ to move towards $D$ (i.e., to reduce their Manhattan distance from $D$). However this naive strategy won't always work. See examples in figure \ref{counter_fig}. In figure \ref{counter_fig_a} the configuration is purely asymmetric and $\lambda_{DA}$ is the strictly largest sequence. So the robots $r_1$ and $r_2$ must move towards $D$. Suppose $r_1$ has crashed, and only $r_2$ moves one step vertically.  The new configuration is given in figure \ref{counter_fig_b}. The configuration is still purely asymmetric, but the largest sequence is now $\lambda_{CD}$. So now $r_2$ and $r_3$ are the leading duo. Suppose they both perceive this configuration, and hence must now move towards $C$. Suppose $r_2$ moves vertically, and $r_3$ horizontally. Assume that they move simultaneously, or synchronously. But this new configuration, shown in figure \ref{counter_fig_c}, has a vertical symmetry and hence partitive. Although $r_1$ has a pending towards $D$, since it has crashed it will not any longer take part in the algorithm. Hence now it has become impossible to gather the robots.

 Hence we need to impose some restrictions on the movement of the robots. But first we formulate a convention to specify the location of a robot in a purely asymmetric configuration. Since in a purely asymmetric configuration there is a unique lexicographically largest sequence, the position of a robot on the grid can be specified uniquely with respect to the largest sequence, in the following way. If $\lambda_{DA}$ is the (strictly) largest sequence, the grid can seen as a square matrix with $DA$ as the first column, and $DC$ as the first row. So if a robot is on the $i$th row and $j$th column, its position can be specified by the tuple $(i, j)$. With this convention in mind, we propose algorithm \ref{toy} for the movement of the leading duo. Incorporating this algorithm in the function  \textsc{Move0}, however, shall require some more work.
 
 \begin{algorithm}[h!]
    \setstretch{0.9}
    \SetKwInOut{Input}{Input}
    \SetKwInOut{Output}{Output}
    \SetKwProg{Fn}{Function}{}{}
    \SetKwProg{Pr}{Procedure}{}{}

    \If{the configuration is purely asymmetric}{
    
    \If{ $r$ is one of the leading duo at $(i, j)$}{
    
    \uIf{($j = 1$) or ($j = 2$ and $i > 2$) or ($i = n$ and $j \leq \frac{n}{2}$)}{
    
    Move along the column towards the largest corner. \;
    
    }
    
    \Else{
    
    Move along the row towards the largest corner. \;
    
    }
    
    }

    }

    \caption{}
    \label{toy} 
\end{algorithm}

%

\begin{lemma}\label{strictly_large}
 Consider a purely asymmetric configuration, with the lexicographically strictly largest sequence being $\lambda_{DA}$. Then after a move by a robot according to algorithm \ref{toy}, $\lambda_{DA}$ remains strictly larger than $\lambda_{CD}, \lambda_{AD}, \lambda_{CB}, \lambda_{AB}, \lambda_{BA}$ and $\lambda_{BC}$.
\end{lemma}
  
\begin{proof}
 If the configuration is purely asymmetric, then only the leading duo will move. Now if the robot corresponding to the first non-zero term of $\lambda_{DA}$ moves, then $\lambda_{DA}$ will remain strictly largest. So we shall only need to consider the case when only the second robot of the leading duo, say $r$, moves.
 
\underline{$\boldsymbol{\lambda_{DA}}$ \textbf{vs} $\boldsymbol{\lambda_{CD}}$}

Observe that the second robot of the leading duo always moves to an unoccupied node or a singly occupied node. Hence due to the weak multiplicity detection capability the move increases $\lambda_{DA}$. Similarly, if this move is along a row, then $\lambda_{CD}$ would decrease, while if the move is along a column, then $\lambda_{CD}$ increases. So we only need to consider moves along a column. 
 
 Before the move suppose that $r$ was at $(i, j)$. Then in terms of the sequence $\lambda_{DA}$ the robot moves from $n(j-1)+i$th place to $n(j-1)+i-1$th place, and in the sequence $\lambda_{CD}$ it moves from $ni-j+1$th place to $n(i-1)-j+1$th place. Hence the first term of $\lambda_{DA}$ that changes is the $n(j-1)+i-1$th term, and all the terms before it remain unchanged. Similarly the first term of $\lambda_{CD}$ that changes is the $n(i-1)-j+1$th term. Initially we had $\lambda_{DA}^{old} > \lambda_{CD}^{old}$. On the contrary assume that after the move we have $\lambda_{DA}^{new} \leq \lambda_{CD}^{new}$.  In that case, the position of the first term of $\lambda_{DA}$ that changes (increases) must be greater than or equal to the position of the first term  in the sequence $\lambda_{CD}$ that increases. To see this let $\Lambda_{DA}$ and $\Lambda_{CD}$ be the sequences obtained from $\lambda_{DA}$ and $\lambda_{CD}$ by taking only first $n(j-1)+i-1$ terms. Then $\Lambda_{DA}^{old} \geq \Lambda_{CD}^{old}$. If $ n(j-1)+i-1 < n(i-1)-j+1 $, then $\Lambda_{DA}^{new} > \Lambda_{CD}^{new}$. Hence $\lambda_{DA}^{new} > \lambda_{CD}^{new}$.

 Therefore, a necessary condition for $\lambda_{DA}^{new} \leq \lambda_{CD}^{new}$ is,
 $$ n(j-1)+i-1 \geq n(i-1)-j+1 $$
 $$ \Rightarrow n(j-i)+(i+j) \geq 2 $$
%

 It is easy to check that for $j=1$ or $j=2$, $i > 2$ or $i = n$, $j \leq \frac{n}{2}$ this criterion is not met, when $n \geq 4$. Hence, $\lambda_{DA}^{new} > \lambda_{CD}^{new}$.

 \underline{$\boldsymbol{\lambda_{DA}}$ \textbf{vs} $\boldsymbol{\lambda_{AB}}$}

 In this case we only need to consider moves along a row. Now consider the following cases.
 
  \textbf{Case 1:} We first assume that $r$ is the first robot in the sequence $\lambda_{AB}$. Let the first non-zero term of $\lambda_{DA}$ and $\lambda_{AB}$ be the $x$th and $y$th term respectively. We must have $y \geq x$.

   \textbf{Case 1A:} First assume that $x = y$. Then $r$ is not on $DA,$ and hence there is only a single robot on $DA$. Since $\lambda_{DA}$ is larger than $\lambda_{AD}$, $x \leq \frac{n}{2}$. Hence, $y \leq \frac{n}{2}$ and so $r$ must move columnwise. As mentioned earlier there is no need to consider columnwise moves.
   
   \textbf{Case 1B:}  Now suppose that $y \geq x+2$, and $r$ moves along the row to the $y-1$th place of $\lambda_{AB}$. Then $\lambda_{DA}$ remains strictly greater than $\lambda_{AB}$ as $y-1 \geq x+1 > x$.
   
   \textbf{Case 1C:} Let $y = x+1$. If $x < \frac{n}{2}$, then $y \leq \frac{n}{2}$ and $r$ must move columnwise. Also $x \ngtr \frac{n}{2}$, as otherwise $\lambda_{AD}$ becomes larger than $\lambda_{DA}$. So assume that $x = \frac{n}{2}$. Since $r$ is the second robot in the sequence $\lambda_{DA}$, the green and blue region together in figure \ref{color_fig1} has only two robots. Since $\lambda_{DA}$ is strictly larger than $\lambda_{BA}$ the pink and blue region together in figure \ref{color_fig1} also has only the same two robots. Then it is easy to see from figure \ref{color_fig2} that even after the move by $r$, $\lambda_{DA}$ remains strictly greater than $\lambda_{AB}$.

  \textbf{Case 2:} Now suppose that $r$ is the second robot in the sequence $\lambda_{AB}$. Assume that $r$ is at the $x'$th position in $\lambda_{DA}$, and at the $y'$th position in $\lambda_{AB}$. Also assume that the non-zero term, corresponding to the first robot, in the sequences $\lambda_{DA}$ and $\lambda_{AB}$ are at the $x$th and $y$th positions respectively. Then $x \leq y$. If $x < y$, then we are done as after the move by $r$, $\lambda_{DA}$ remains strictly larger than $\lambda_{AB}$. So let $x = y$. In this case we should have $x' \leq y'$. After a rowwise move by $r$, it moves from $x'$th position in $\lambda_{DA}$ to $(x'-n)$th position, and  from $y'$th position in $\lambda_{AB}$ to $(y'-1)$th position. Then even after the move $\lambda_{DA}$ remains strictly larger than $\lambda_{AB}$, as $x'-n < y'-1$.

  \begin{figure}[thb!]
\centering
\subcaptionbox[Short Subcaption]{
       Before the move by $r$ \label{color_fig1}
}
[
    0.48\textwidth 
]
{
    \def\svgwidth{0.48\textwidth}
    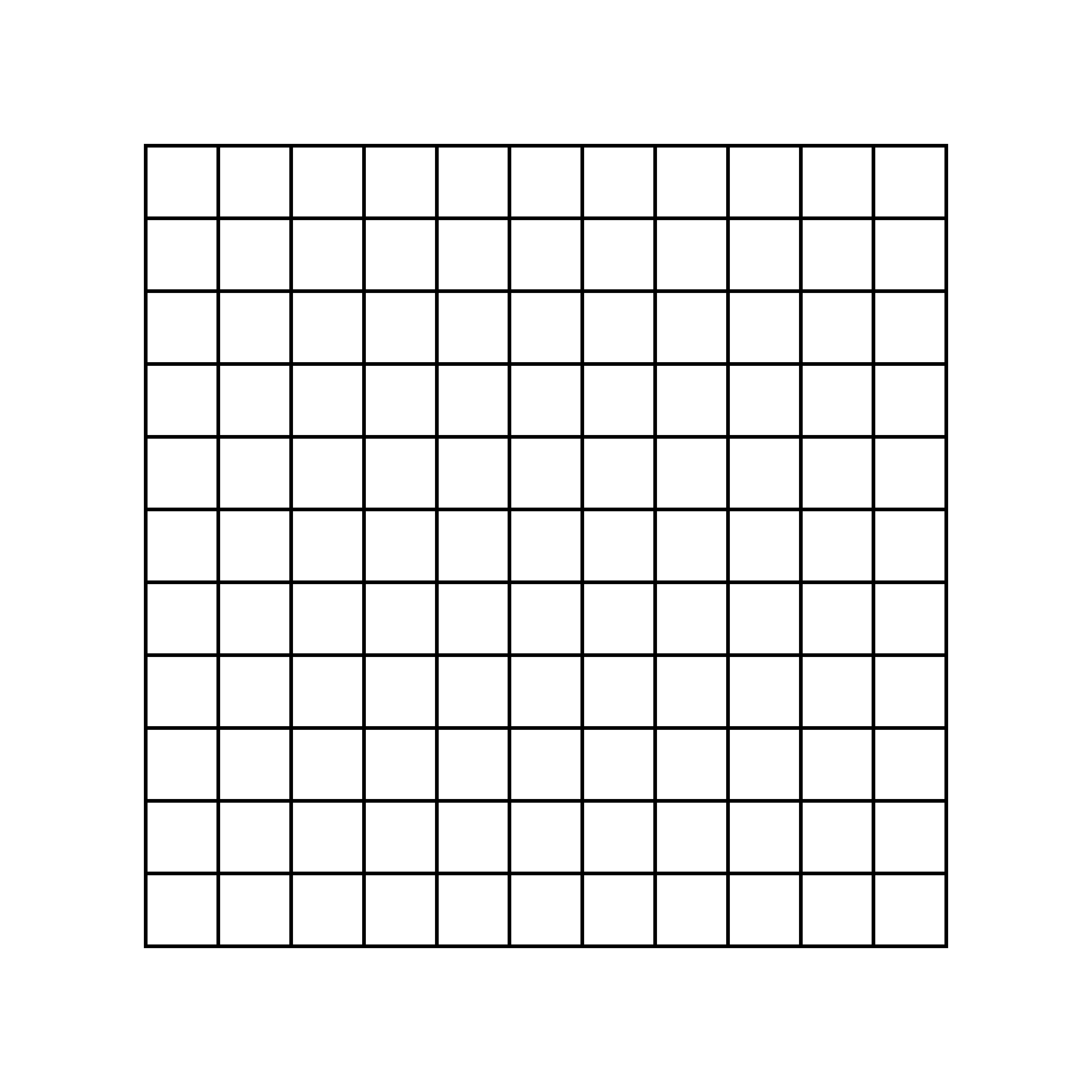
}
\hfill 
\subcaptionbox[Short Subcaption]{
   After the move by $r$ \label{color_fig2}
}
[
    0.48\textwidth 
]
{
    \def\svgwidth{0.48\textwidth}
    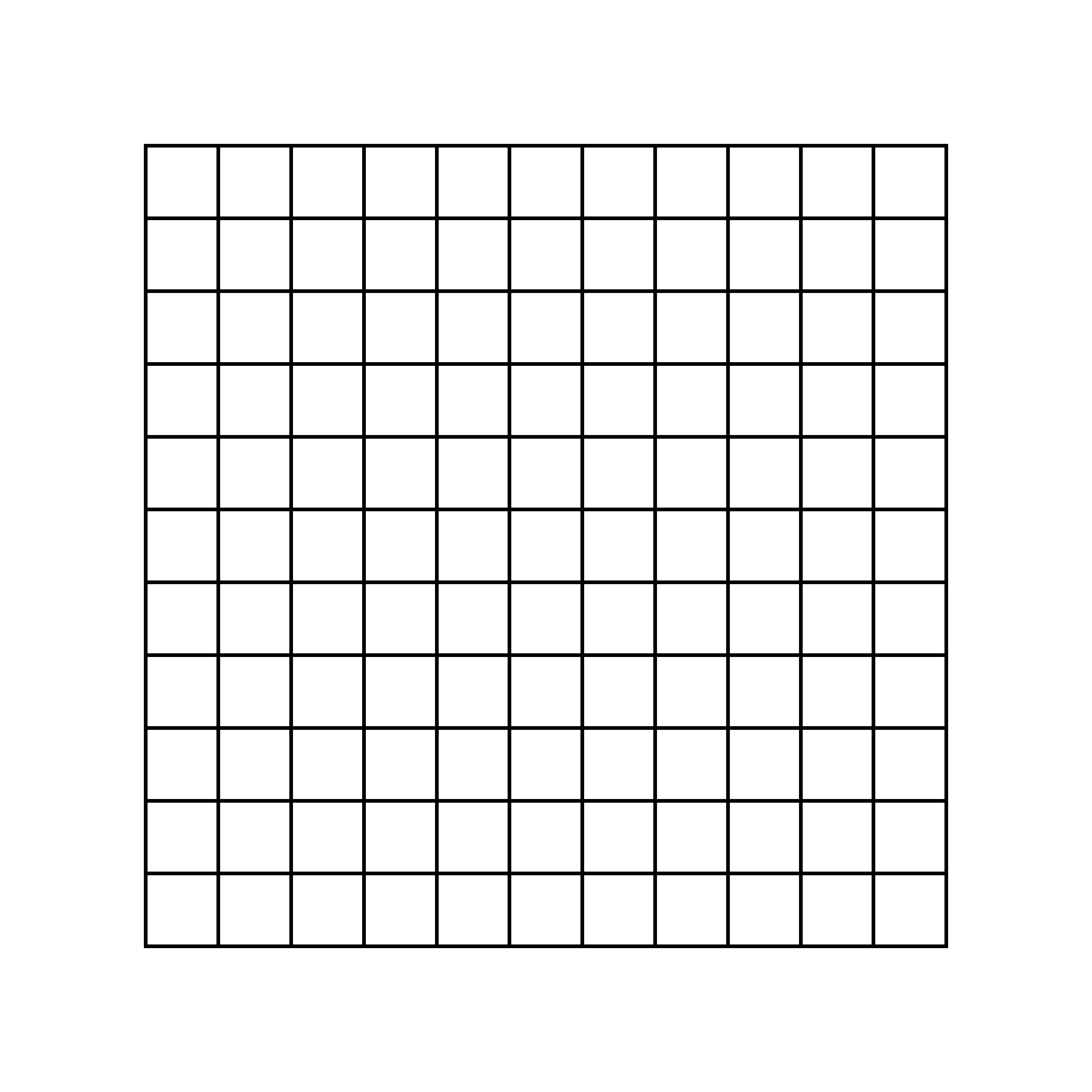
}
\caption[Short Caption]{Illustrations supporting the proof of case 1C of lemma \ref{strictly_large}}
\label{color_fig}
\end{figure}

  \textbf{Case 3:} Finally consider the case where $r$ neither the first nor the second robot in the sequence $\lambda_{AB}$. But since $r$ is the second robot in the sequence $\lambda_{DA}$, after the move $\lambda_{DA}$ remains strictly larger than $\lambda_{AB}$.

 \underline{$\boldsymbol{\lambda_{DA}}$ \textbf{vs} $\boldsymbol{\lambda_{AD}}$}
 
 Again we only need to consider moves along a row. Hence assume that $r$ is not in $AD$. This implies that there is only one robot, say $r'$, in $AD$. Suppose that $r'$ is at the $x$th position in the sequence $\lambda_{DA}$. Since $\lambda_{DA}$ is strictly larger than $\lambda_{AD}$, $x \leq \frac{n}{2}$. So clearly if $r$ does not move to $AD$, then $\lambda_{DA}$ remains strictly larger than $\lambda_{AD}$, as it is decided looking at first $\frac{n}{2}$ terms. Now if $r$ moves to $AD$, then it moves either to (1,1) or (2,1), according to algorithm \ref{Move0}. Clearly in either case, $\lambda_{DA}$ remains strictly larger than $\lambda_{AD}$.

 \underline{$\boldsymbol{\lambda_{DA}}$ \textbf{vs} $\boldsymbol{\lambda_{CB}}$}
 
 We only need to consider columnwise moves. Let $(i,j)$ be the position of $r$ before the move. Then in terms of the sequence $\lambda_{DA}$ the robot moves from $n(j-1)+i$th place to $n(j-1)+i-1$th place, and in the sequence $\lambda_{CB}$ it moves from $n(n-j)+i$th place to $n(n-j)+i-1$th place. Using the same arguments as earlier, we have the following sufficient condition for ${\lambda_{DA}}$ to remain strictly larger than ${\lambda_{CB}}$:
 $$n(j-1)+i-1 < n(n-j)+i-1$$
 $$\Rightarrow j \leq \frac{n}{2}$$
 
Clearly ${\lambda_{DA}}$ remains strictly larger, since $r$ does not move columnwise when $j > \frac{n}{2}$

\underline{$\boldsymbol{\lambda_{DA}}$ \textbf{vs} $\boldsymbol{\lambda_{BC}}$ \textbf{and} $\boldsymbol{\lambda_{DA}}$ \textbf{vs} $\boldsymbol{\lambda_{BA}}$}

Since $r$ moves towards $D$, it actually moves away from $B$. Hence even after the move ${\lambda_{DA}}$ remains strictly larger than both $\lambda_{BA}$ and $\lambda_{BC}$.  $\square$
 
\end{proof}

 \begin{lemma}\label{DC}
 Consider a purely asymmetric configuration, with the strictly largest sequence being $\lambda_{DA}$. Let $\{r_1, r_2\}$ be the leading duo with $r_1$ being the first and $r_2$ the second robot. Then after a move by $r_1$ or $r_2$ or both, according to algorithm \ref{toy}, the new configuration is one of the following:
 
 \begin{enumerate}
  \item purely asymmetric with a possible pending move consistent with algorithm \ref{toy},
  
  \item symmetric of the first type with a possible pending move of $r_1$ towards the largest corner $D$,
  
  \item almost symmetric of the first type with a possible pending move of $r_1$ towards the largest corner $D$,
  
  \item almost symmetric of the second type with a possible pending move of $r_1$ towards  $D$.
 \end{enumerate}

\end{lemma}

 \begin{proof}
   
  If only  $r_1$ moves then clearly the configuration remains asymmetric. We show that the new configuration is also not almost symmetric. To see this assume that $r_1$ was initially at the $x$th node from $D$ on $DA$. After a move the configuration can not become almost symmetric of the first type as that would imply that the initial configuration was also almost symmetric of the first type. Now assume that the new configuration is almost symmetric of the second type. If $C$ is the second largest corner, then it implies that there is a robot on $AB$ at $(x-1)$th place from $B$. The same would imply if $A$ is the second largest corner and $r_1$ is not impeding symmetry. So the remaining case is that $A$ is the second largest corner and $r_1$ is impeding symmetry, i.e., $r_1$ is the only robot on $DA$. But that would imply that the initial configuration was also almost symmetric of the second type.

  So assume that $r_2$ also moves. After a move by $r_2$, $\lambda_{DA}$ remains to be strictly larger than $\lambda_{CD}, \lambda_{AD}, \lambda_{CB}, \lambda_{AB}, \lambda_{BA}$ and $\lambda_{BC}$. Hence after the move the configuration does not admit any partitive symmetry. However, after such a move $\lambda_{DA}$ and $\lambda_{DC}$ may become equal. Hence the configuration may become symmetric of the first type. In this case $r_1$ has not moved, and may have a pending move towards $D$. A move by $r_2$ can also make the configuration almost symmetric of the first type or almost symmetric of the second type. Again $r_1$ can have a pending move towards $D$. In the case when the new configuration is almost symmetric of the first type, the pending move is consistent with the algorithm. But when the new configuration is almost symmetric of the second type, the pending move is conflicting with the algorithm for almost symmetric configurations of the second type. Also note that a move by $r_2$ can not make the configuration symmetric of the second type. This is because before the move we had $\lambda_{DA}^{old} > \lambda_{BA}^{old}$. Now after a move by $r_2$ towards $D$ (and thus away from $B$), we shall have $$\lambda_{DA}^{new} > \lambda_{DA}^{old} > \lambda_{BA}^{old} > \lambda_{BA}^{new}$$
  $\square$
  \end{proof}

  Thus in a purely asymmetric configuration, a move by the second robot of the leading duo may create an almost symmetric configuration of the second type. In this case there might be a pending move which conflicts with the algorithm for almost symmetric configurations of the second type. This pending move can break down the entire gathering algorithm. To see this consider the example in figure \ref{counter_critical}. The configuration in figure \ref{counter_critical_a} is purely asymmetric, with $\lambda_{DA}$ being the largest sequence and $\{r_1, r_2\}$ the leading duo. Now according to algorithm \ref{toy}, $r_2$ moves columnwise towards $D$. But $r_1$ is yet to move, and has a pending move towards $D$. The new configuration in figure \ref{counter_critical_b} is almost symmetric of the second type, and $A$ is the second largest corner. Here $r_1$ and $r_3$ are the robots impeding symmetry, and hence should go towards $A$. But $r_1$ has a move pending towards $D$. Now suppose $r_3$ makes a move towards $A$. $r_1$ still has a move pending towards $D$. Now suppose $r_1$ and $r_3$ move simultaneously towards $D$ and $A$ respectively. The resulting configuration (figure \ref{counter_critical_d}) is a partitive configuration with no pending moves,
  and hence it is now impossible to gather the robots.
  
  \begin{figure}[thb!]
\centering
\subcaptionbox[Short Subcaption]{
       \label{counter_critical_a}
}
[
    0.48\textwidth 
]
{
    \def\svgwidth{0.48\textwidth}
    \import{}{images/counter_critical_a.pdf_tex}
}
\hfill 
\subcaptionbox[Short Subcaption]{
     \label{counter_critical_b}
}
[
    0.48\textwidth 
]
{
    \def\svgwidth{0.48\textwidth}
    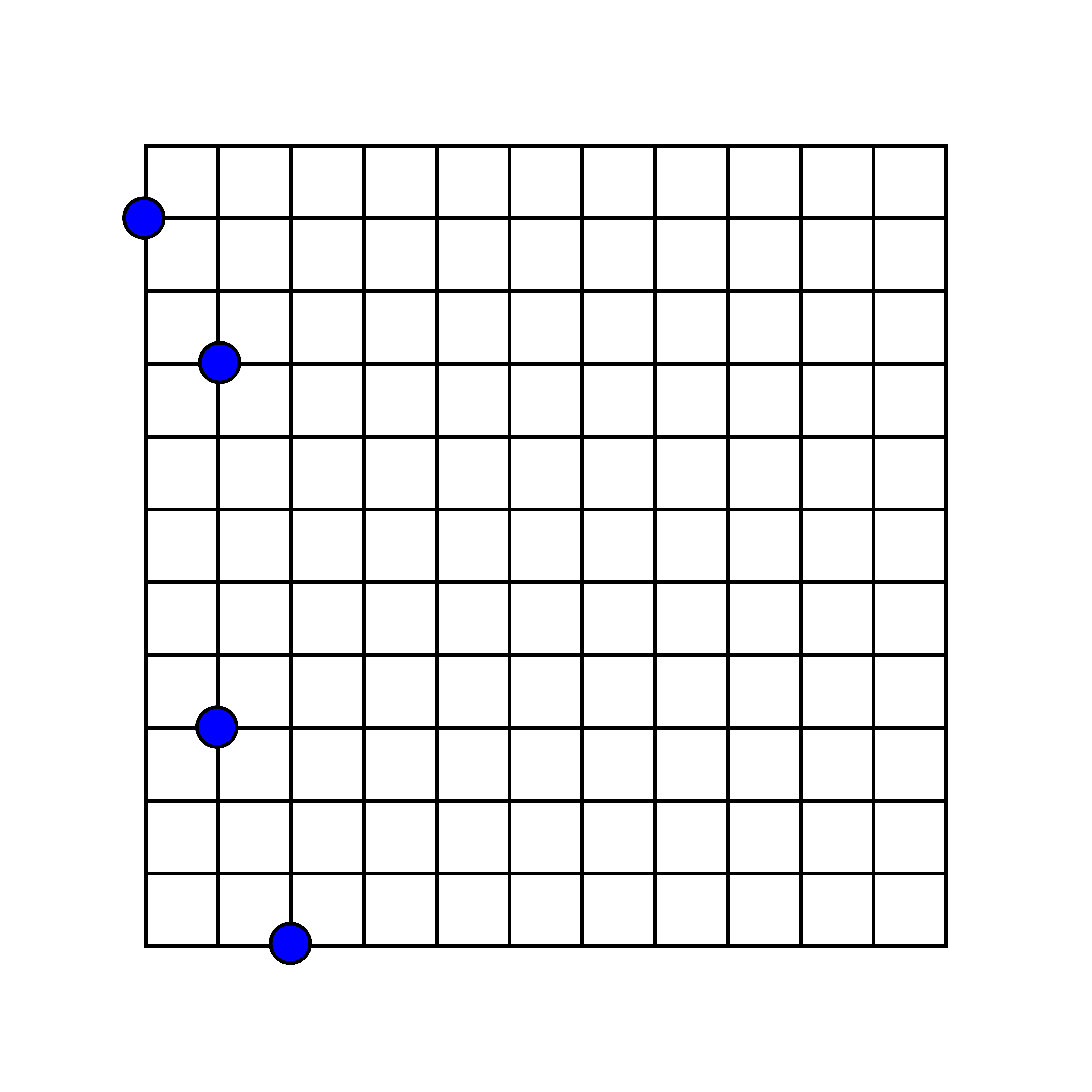
}
\\
\subcaptionbox[Short Subcaption]{
        \label{counter_critical_c}
}
[
    0.48\textwidth 
]
{
    \def\svgwidth{0.48\textwidth}
    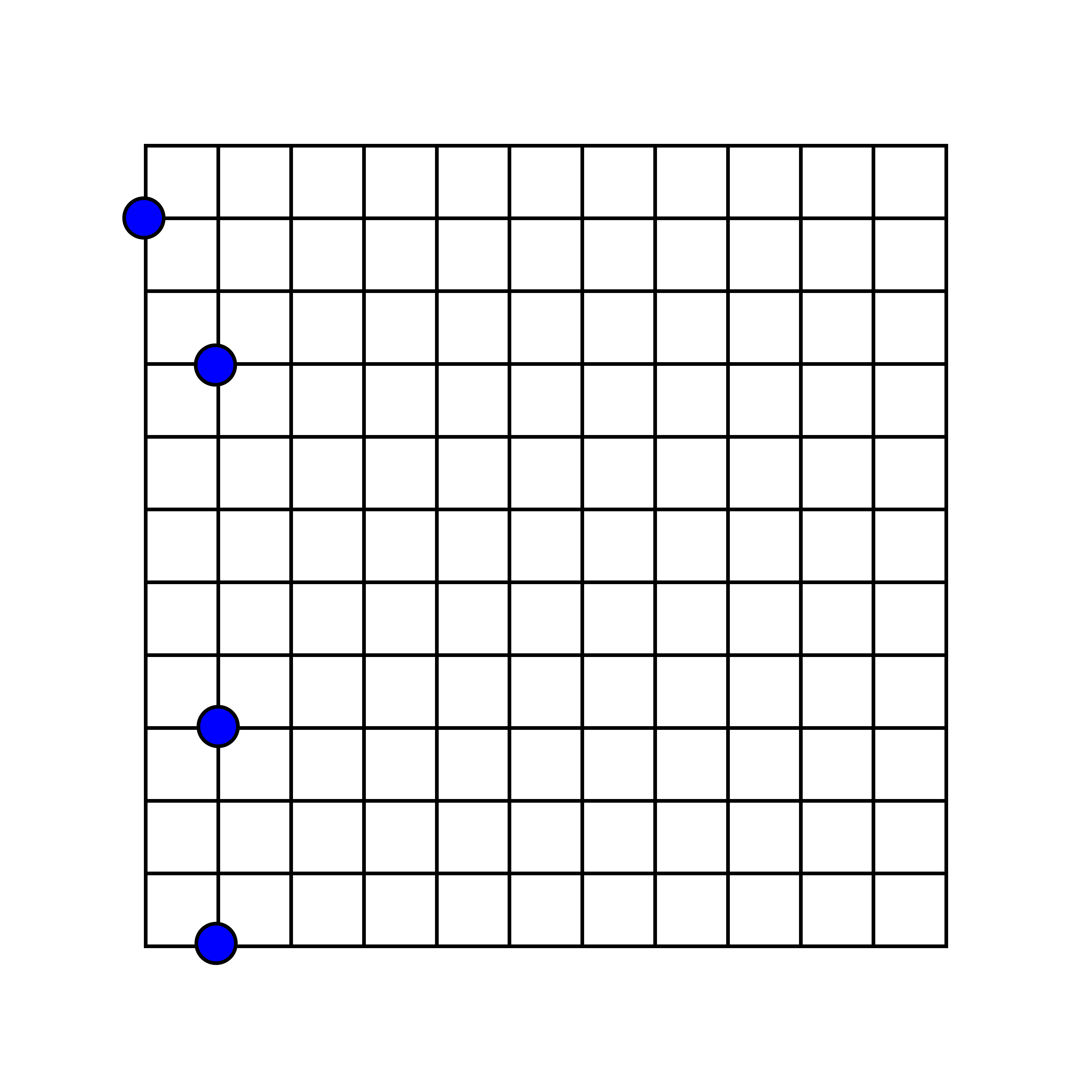
}
\hfill
\subcaptionbox[Short Subcaption]{
       \label{counter_critical_d}
}
[
    0.45\textwidth 
]
{
    \def\svgwidth{0.45\textwidth}
    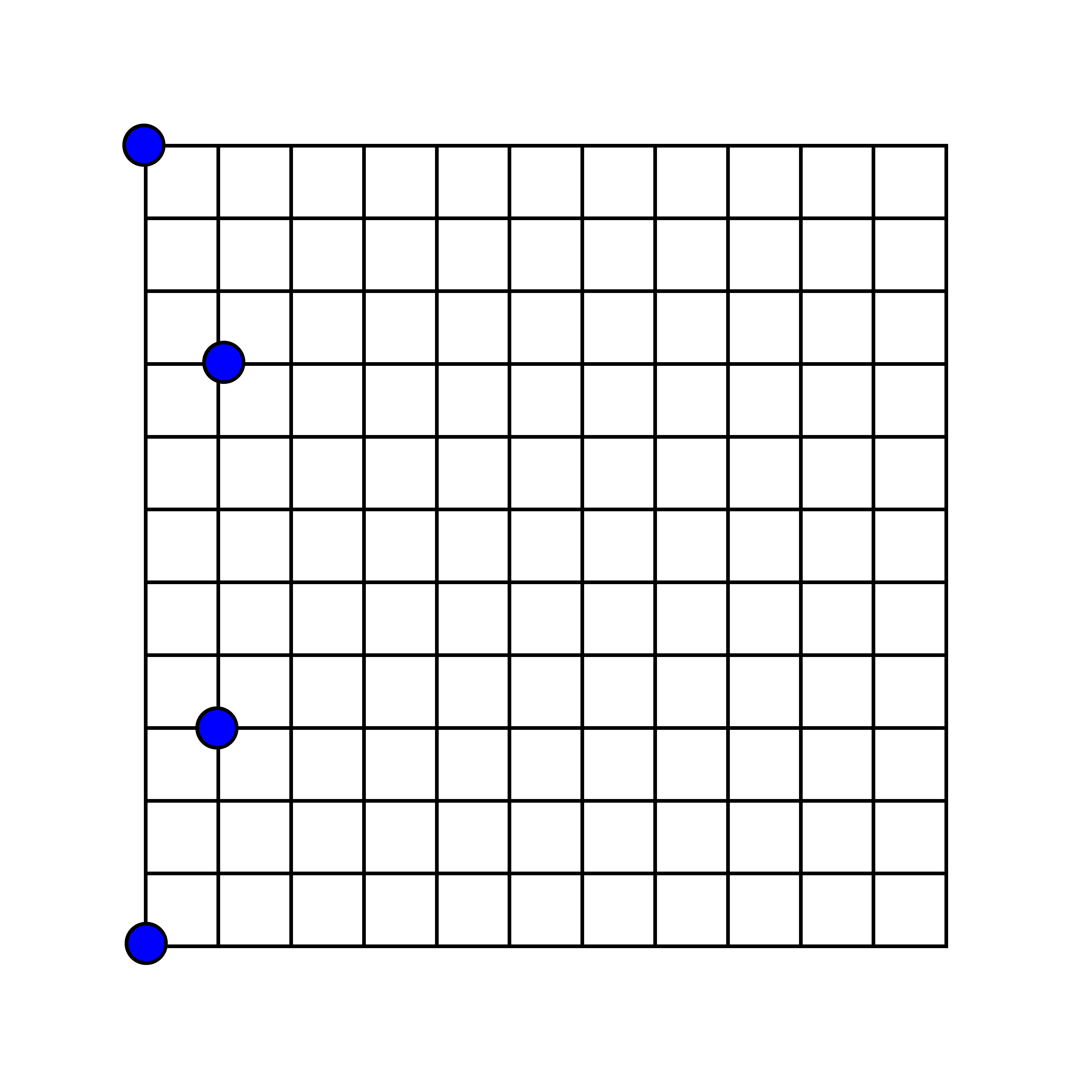
}

\caption[Short Caption]{Algorithm \ref{toy} can not be followed in critical configurations}
\label{counter_critical}
\end{figure}

Hence when the initial configuration is purely asymmetric, we have to make sure that the configuration doesn't become almost symmetric of the second type. The second robot of the leading duo can easily check if its move according to algorithm \ref{toy} will make the configuration almost symmetric of the second type. We shall call these configurations \emph{critical configurations}. Whenever the configuration becomes critical, the second robot of the leading duo will not move. Instead we have to ask some other robot to move. First we need to prove some results regarding the critical configuration.

\begin{lemma}\label{critical}
 Consider a critical configuration $\mathcal{C}$ with $\lambda_{DA}$ the lexicographically largest sequence. Assume that $\{r_1, r_2\}$ are the leading duo with $r_1$ being the first and $r_2$ the second robot. Let $\mathcal{C}'$ be the configuration after a move by $r_2$ according to the algorithm \ref{toy}. (Clearly according to the definition of critical configuration, $\mathcal{C}'$ is almost symmetric of the second type.) Then we have the following:
 
 \begin{enumerate}
  \item in $\mathcal{C}$, $r_2$ is not on $DA,$
  \item in $\mathcal{C}'$, $r_2$ is not on $DA,$
  \item in $\mathcal{C}$, $DA$ has exactly one robot, $r_1$, and $BA$ also has exactly one robot, say $r_3$,
  \item $A$ is the second largest corner in $\mathcal{C}'$.
 \end{enumerate}

\end{lemma}

\begin{proof}
 \textbf{1.} This follows from 2, since if $r_2$ is initially on $DA$ then after the move according to algorithm \ref{toy}, it remains on $DA$. Hence we only need to prove 2.

 \textbf{2.} On the contrary assume that after the move, $r_2$ reaches $DA$. Assume that in $\mathcal{C}'$, $r_1$ and $r_2$ are at $x$th and $y$th place from $D$ on $DA,$ respectively. Clearly $x \neq y$, because an almost symmetric configuration of the second type has no multiplicities.

  \textbf{Case 1 ($A$ is the second largest corner): } Consider the following cases.

  \textbf{Case 1A ($\boldsymbol{x>y}$):} Clearly $r_2$ is not a robot impeding symmetry on $DA$ in $\mathcal{C}'$. This implies that there is a robot on $AB$ at $y$th place from $B$. Hence $\lambda_{DA}^{old} < \lambda_{BA}^{old}$, a contradiction.
   
   \textbf{Case 1B ($\boldsymbol{x<y}$ and $\boldsymbol{r_2}$ is not a robot impeding symmetry):} This implies that there is a robot on $DA$ in $\mathcal{C}'$ beyond the $y$th place from $D$. Hence in  $\mathcal{C}$, $r_2$ was on $DA$  at $(y+1)$th place from $D$. Also there must be a robot at $y$th place from $B$ on $BA$ in $\mathcal{C}$. But the $y$th place from $D$ on $DA$  in $\mathcal{C}$ is empty. So comparing the first $y$ terms of $\lambda_{DA}^{old}$ and $\lambda_{BA}^{old}$, we can conclude that $\lambda_{DA}^{old} < \lambda_{BA}^{old}$, which is a contradiction.
   
   \textbf{Case 1C ($\boldsymbol{x<y}$ and $\boldsymbol{r_2}$ is a robot impeding symmetry):} If $r_2$ is on $DA$ in $\mathcal{C}$, then the configuration is already an almost symmetric of the second type. Hence $r_2$ is not on $DA$ in $\mathcal{C}$. This implies that in $\mathcal{C}$, there is only $r_1$ on $DA,$ which is at $x$th place from $D$. There is a robot $r_3$ on $BA$ at $x$th place from $B$. Now $r_3$ is not impeding symmetry on $BA$ in $\mathcal{C}'$, as $r_1$ is not impeding symmetry on $DA$. Hence there is another robot on $BA$ other than $r_3$. Now compare the first $n$ terms of $\lambda_{DA}^{old}$ and $\lambda_{BA}^{old}$. In case of $\lambda_{DA}^{old}$, the $x$th term is 1 and the other $n-1$ terms are all 0. For $\lambda_{BA}^{old}$, the $x$th term is also 1, but there is another 1 in the first $n$ terms. This implies that $\lambda_{DA}^{old} < \lambda_{BA}^{old}$, a contradiction.

  \textbf{Case 2 ($C$ is the second largest corner):} Construct sequences  $\Lambda_{DA}$ and $\Lambda_{BA}$ from  $\lambda_{DA}$ and $\lambda_{BA}$ respectively, by taking only the first $n$ terms. We must have $\Lambda_{DA}^{old} \geq \Lambda_{BA}^{old}$. After $r_2$ has moved on or to $DA,$ obviously we have $\Lambda_{DA}^{new} > \Lambda_{BA}^{new}$. But if the new configuration $\mathcal{C}'$ is almost symmetric of the second type, with $C$ as the second largest corner, then $\Lambda_{DA}^{new} = \Lambda_{BA}^{new}$, as the impeding robots are on $DC$ and $BC$. Thus we have a contradiction.

\textbf{3.} By 1, $r_2$ is not on $DA$ in $\mathcal{C}$. But it is the second robot. Hence $r_1$ is the only robot on $DA$ in $\mathcal{C}$. So it remains to show that $BA$ has only one robot in $\mathcal{C}$. To do that consider the following cases.

 \textbf{Case 1 ($A$ is the second largest corner):} By 2, $DA$ has only one robot in $\mathcal{C}'$, which is $r_1$. Hence if $A$ is the second largest corner, $r_1$ is a robot impeding symmetry in $\mathcal{C}'$. Then $BA$ also has only one robot in $\mathcal{C}'$. Hence in $\mathcal{C}$, $BA$ has at most two robots. If it has only one robot, then we are done. So assume that there are exactly two robots on $BA$ in $\mathcal{C}$. In that case the robot on $BA$ closest to $A$ is $r_2$ and it has made a columnwise move. Suppose that $r_1$ is on $DA$ at $x$th place from $D$ and $r_2$ is on $BA$ at $y$th place from $A$. First let $y > 2$. Then as $r_2$ is the second robot, the second column is empty in $\mathcal{C}$. After the move $r_2$ is now on the $(n-1)$th row, while the second column is still empty. This clearly contradicts the fact that $\mathcal{C}'$ is almost symmetric of the second type. So we must have $y=2$. So in $\mathcal{C}$, $BA$ has two robots with one of them very next to the corner $A$. Then clearly comparing the first $n$ terms we see that $\lambda_{DA}^{old} < \lambda_{AB}^{old}$, a contradiction.
 
 \textbf{Case 2 ($C$ is the second largest corner):} In this case the first $n$ terms of $\lambda_{DA}^{new}$ and $\lambda_{BA}^{new}$ are equal. By 2, the first $n$ terms of $\lambda_{DA}^{new}$ has exactly one 1, say at $x$th palce, and the rest are 0. Then in $\mathcal{C}'$, and hence also in $\mathcal{C}$, there is a robot on $BA$ at $x$th place from $B$. Since $\lambda_{DA}^{old}$ is strictly largest, there is no other robot on $BA$ in $\mathcal{C}$.

\textbf{4.} Lastly we show that $A$ is the second largest corner. On the contrary assume that $C$ is the second largest corner. Before the move, let $r_2$ correspond to the $y$th term in $\lambda_{DA}^{old}$. Now construct sequences  $\Lambda_{DA}$ and $\Lambda_{BA}$ from  $\lambda_{DA}$ and $\lambda_{BA}$ respectively, by taking only the first $y$ terms. We must have $\Lambda_{DA}^{old} \geq \Lambda_{BA}^{old}$. After a move by $r_2$, we have $\Lambda_{DA}^{new} > \Lambda_{BA}^{new}$. Let $\tilde{\Lambda}_{DA}^{new}$ and $\tilde{\Lambda}_{BA}^{new}$ be obtained from $\Lambda_{DA}$ and $\Lambda_{BA}$ by replacing any 1 with 0, if it corresponds to a robot impeding symmetry in $\mathcal{C}'$. Note that a robot impeding symmetry can only be on $DC$ or $BC$. We claim that $\tilde{\Lambda}_{DA}^{new} = {\Lambda}_{DA}^{new}$. $\Lambda_{DA}^{new}$ has only two non-zero terms, a 1 corresponding to $r_1$ and another 1 corresponding to $r_2$. Of these two only $r_2$ can be a robot impeding symmetry in $\mathcal{C}'$, and in that case it must be on $DC$. But in that case the configuration was almost symmetric of the second type even before the move. This implies that $\tilde{\Lambda}_{DA}^{new} = {\Lambda}_{DA}^{new}$. Thus we have

$$\tilde{\Lambda}_{DA}^{new} = {\Lambda}_{DA}^{new} > \Lambda_{BA}^{new} \geq  \tilde{\Lambda}_{BA}^{new}$$
$$ \Rightarrow  \tilde{\Lambda}_{DA}^{new} > \tilde{\Lambda}_{BA}^{new}$$.

This is a contradiction, since if $\mathcal{C}'$ is almost symmetric of the second type, then $\tilde{\Lambda}_{DA}^{new} = \tilde{\Lambda}_{BA}^{new}$. Hence $C$ is not the second largest corner. $\square$
\end{proof}

Hence if a configuration is critical with $\lambda_{DA}$ as the largest corner, then by lemma \ref{critical} the boundary sides $AD$ and $AB$ contain a single robot each. In a critical configuration, these two robots will be the leading duo.

\begin{lemma}\label{critical2}
 Consider a critical configuration with $\lambda_{DA}$ as the largest corner. Then by lemma \ref{critical}, there is a single robot on $AB$, say $r_3$. If it makes a columnwise move then $\lambda_{DA}$ remains the strictly largest sequence.  
\end{lemma}

\begin{proof}
 
 We only need to prove that after the move $\lambda_{DA}$ remains strictly larger than $\lambda_{CB}$, $\lambda_{CD}$ and $\lambda_{DC}$. Assume that the first robot of the leading duo, $r_1$, is on $DA$ at $x$th place from $D$ and $r_3$ on $AB$ is at $y$th place from $A$. Also let $r_2$ be the second robot in $\lambda_{DA}$.
%

 \underline{$\boldsymbol{\lambda_{DA}}$ \textbf{vs} $\boldsymbol{\lambda_{CB}}$}
 
 Let the robot on $CB$ closest to $C$ is at $z$th place from $C$. Let $x < y$. Then as $A$ is the second largest corner, $z \geq y$. Hence $z > x$. So even after the move by $r_3$, $\lambda_{DA}$ remains strictly larger than $\lambda_{CB}$. So now assume that $x = y = z$. Now $x \leq \frac{n}{2}$, for otherwise $\lambda_{DA} < \lambda_{AD}$. Hence also $y \leq \frac{n}{2}$.  Now construct sequences  $\Lambda_{DA}$ and $\Lambda_{CB}$ from  $\lambda_{DA}$ and $\lambda_{CB}$ respectively, by taking only the first $\frac{n^2}{2}$ terms. Then it easy to see that $$ \Lambda_{DA}^{new} > \Lambda_{DA}^{old} \geq \Lambda_{CB}^{old} = \Lambda_{CB}^{new} $$.
 This implies that $\lambda_{DA}^{new} > \lambda_{CB}^{new}$.
 
 \underline{$\boldsymbol{\lambda_{DA}}$ \textbf{vs} $\boldsymbol{\lambda_{CD}}$}
 
 It is similar to the last proof.
 
 \underline{$\boldsymbol{\lambda_{DA}}$ \textbf{vs} $\boldsymbol{\lambda_{DC}}$}
 
 Let $r_3$ be at $(n, j)$. It is easy to see that if $j < n-1$, then $\lambda_{DC}$ can not overtake $\lambda_{DA}$. So let $j = n-1$. Since $\lambda_{DA}$ is strictly largest, $r_1$ is at $(2, 1)$. By lemma \ref{critical}, if $r_2$ moves then $A$ will be the second largest corner and $r_1, r_3$ will be the robots impeding symmetry in the new configuration which is almost symmetric of the second type. But this is impossible since $r_1, r_3$ are in symmetric positions with respect to the diagonal axis $AC$. $\square$
 
\end{proof}

\begin{theorem}\label{move1th}
 If the initial configuration is purely asymmetric, then algorithm \ref{Move0} leads to a configuration with exactly one corner occupied.
\end{theorem}

\begin{proof}
 Let $\lambda_{DA}$ be the strictly largest sequence. Suppose that the configuration is not critical. After the move by the leading duo, by lemma \ref{DC} the new configuration is purely asymmetric or symmetric of the first type or almost symmetric of the first type. In each case any pending move is consistent with the algorithm. If the new configuration is symmetric of the first type or almost symmetric of the first type, then by theorem \ref{second_type_theorem} eventually exactly one corner gets occupied. If the new configuration is purely asymmetric and not critical, then again the same moves are to be executed by the leading duo. So now assume that the initial configuration is critical or a critical configuration is obtained from a purely asymmetric configuration. Let $\{r_1, r_3\}$ be the leading duo, where $r_1$ is on $DA$ and $r_3$ is on $AB$. If the critical configuration is obtained from a purely asymmetric configuration, then $r_1$ may have a pending move towards $D$. But that is consistent with the algorithm. If $r_3$ moves, then new configuration is no longer critical as $AB$ is empty. Also $\lambda_{DA}$ remains strictly largest by lemma \ref{critical2}, and the leading duo is $\{r_1, r_2\}$, $r_2$ is the second robot in $\lambda_{DA}$. Furthermore, from this configuration a critical configuration can never be created as $AB$ has become empty and will remain so. However $r_3$ may crash and never move from $AB$. But then $r_1$ will keep moving towards $D$ (as there can be at most one crash), and eventually $D$ gets occupied. $\square$
\end{proof}

The following result immediately follows from theorems \ref{second_type_theorem}, \ref{move1th} and \ref{move0th}.

\begin{theorem}
 A non-partitive configuration on an even $\times$ even square grid with no corners occupied is gatherable despite at most one crash fault.
\end{theorem}

\subsubsection{Exactly two corners occupied}
 
Consider the case where exactly two corners of the grid are occupied. The two occupied corners can either be diagonally opposite or the two end-points of a side of the grid. If the two occupied corners lie on the same side of the grid, then the configuration can not have a diagonal symmetry. So in this case a gatherable configuration must be asymmetric. But if the two occupied corners are diagonally opposite, the configuration can have a diagonal symmetry. In that case the axis of symmetry passes through either the two occupied corners or the two unoccupied corners. In the later case the configuration will be called a \emph{2S2 configuration}.

When the two occupied corners are lying on the same side of the grid, each occupied corner has two sequences attached to it: one has 1 at the $n$ place, while the other has 0 at the $n$th place. We shall only consider the later one. The two sequences associated with the two occupied corners can not be equal, for otherwise the configuration becomes partitive. The occupied corner with the lexicographically larger sequence will be called the \emph{larger occupied corner}.  The other occupied corner will be the \emph{smaller occupied corner}. For the case with two diagonally opposite occupied corners, if the configuration is not 2S2, the larger occupied corner can be similarly defined by comparing the four sequences associated with the occupied corners. The other occupied corner will again be called the smaller occupied corner.

\begin{theorem}\label{th2}
 If the initial configuration on an even $\times$ even square grid with exactly two corners occupied is non-partitive and not 2S2, then it is gatherable despite at most one crash fault.
\end{theorem}

\begin{proof}

 In each case the following robots will be asked to move towards the larger occupied corner: 1) the robot at the smaller occupied corner, and 2) the robot that corresponds to the leading non-zero term (not corresponding to an occupied corner) of the largest sequence. If the robot at the smaller occupied corner moves then there is only one occupied corner and all the robots will start moving towards it. If the robot corresponding to the leading non-zero term of the largest sequence moves then the larger occupied corner remains unchanged. So even if the robot at the smaller occupied corner crashes, a multiplicity will be created at the larger occupied corner after finite time. Once a multiplicity is created at a corner, all the robots will be asked to move towards it. $\square$

%
%
%

\end{proof}

\subsubsection{Exactly three corners occupied}

Now we consider configurations with exactly three corners of the grid occupied. We shall call the occupied corner diagonally opposite to the unoccupied corner as the \emph{angular corner}.

\begin{theorem}\label{th3}
 All configurations on an even $\times$ even square grid with exactly three corners occupied are gatherable despite at most one crash fault.
\end{theorem}

\begin{proof}

  \textbf{Case 1:} Suppose that there are at least 5 robots. Then there are at least 2 robots other than the ones at the corners. We shall ask these robots to go to the angular corner. Hence despite a crash fault, a multiplicity will be created at the angular corner after finite time. Then the rest will move to that corner.

  \textbf{Case 2:} Now assume that there are exactly four robots. Let $A$ be the angular corner. Let $r_1, r_2, r_3$ be the robots at the corners $D, A$ and $B$ respectively. Let $r_4$ be the remaining robot.

  \textbf{Case 2A:} Let $r_4$ be on $DC$ or $BC$. Without loss of generality assume that it is on $DC$. Then $r_2$ and $r_4$ will be asked to move towards $D$. If $r_2$ moves then only two corners of the grid become occupied, with $D$ being the larger occupied corner. So a possible pending move by $r_4$ is consistent with the algorithm described in theorem \ref{th2} and both robots will continue moving towards $D$ accordingly. If $r_2$ crashes and doesn't leave $A$, then $r_4$ will reach $D$. After a multiplicity is created at $D$, $r_3$ will move towards it. So in any case eventually all non-faulty robots will meet at $D$.
  
  \textbf{Case 2B:} Suppose that $r_4$ is not on $DC$ or $BC$. Then we ask $r_1$ and $r_3$ to move towards $A$. If both $r_1, r_3$ move simultaneously towards $A$, then the new configuration has exactly one corner occupied and we are done. So now assume that only $r_3$ moves. After the move the new configuration has exactly two corners occupied. Note that regardless of where $r_4$ is, after the move by $r_3$, $A$ is the largest corner. So a possible pending move by $r_1$ is consistent with the algorithm described in theorem \ref{th2}. Thus eventually all the non-faulty robots will gather at $A$.

  \textbf{Case 3:} Lastly assume that the initial configuration has exactly 3 robots all at distinct corners. Then the robots will be asked to move towards the angular corner. Gathering will be achieved like case 2B. $\square$

\end{proof}

\subsubsection{Exactly four corners occupied}\label{sec4}

Similar to the case where no corners were occupied, we can classify the non-partitive configurations as asymmetric, symmetric of the first type and symmetric of the second type. If there are at least two robots other than the ones at the corners (which means that the configuration has at least 6 robots) and at least one of them on the boundary, then we can replicate the algorithm for no corners occupied. In that case the aim would be to create a multiplicity at one of the corners, and then gather all the remaining robots at that corner. So we shall consider exclusively the configurations with either exactly 5 robots or that has no robots on the boundary except at the corners.

First consider configurations with exactly 5 robots. Then the configuration is either asymmetric or symmetric of the first type. Therefore let $D$ be the strictly largest corner. Then the robot at $B$ and the one not at any corner will be asked to move towards $D$. Now consider configurations with at least 6 robots and no robots on the boundary except at the corners. If the configuration is asymmetric or symmetric of first type with $D$ being the largest corner, then the robots on $A$ and $C$ will be asked to move towards $D$. If the configuration is symmetric of second type with $A$ being the second largest corner, the robots on $D$ and $B$ will move towards $A$.

Notice that this algorithm may create configurations with exactly three corners occupied. In that case if the robots start behaving according to the algorithm described in theorem \ref{th3}, then gathering may not be achieved due to conflicting pending moves. Hence we need different strategy for configurations with exactly three corners occupied that may arise from the algorithm we just described. The following are the two configurations with three occupied corners that can arise from a configuration with four corners occupied.  Assume that the corners $D, A$ and $B$ are singly occupied.

\textbf{3C1 configuration:} Consider a configuration with exactly 4 robots or at least 6 robots having exactly four singleton robots on the boundary of the grid: three at the three corners, and one on a side of the grid that has the unoccupied corner as an end-point, i.e. $CD$ or $CB$. This configuration will be called a \emph{3C1 configuration}. Without loss of generality assume that the fourth robot is on $CD$. This robot will be called the \emph{solitary robot}. In this case the solitary robot and the robot at $A$ will be asked to move towards $D$. Note that a 3C1 configuration with exactly four robots can not arise from an initial configuration with four occupied corners, because a configuration with exactly fours robots at fours corners is ungatherable. However we have included the 4 robot case in the definition of 3C1 as it will be useful later on in the unified algorithm presented in section \ref{sec_uni}.

 \textbf{3C2 configuration:} If there are exactly 5 robots in the configuration, then it will be called a \emph{3C2 configuration}. Then the two robots not at a corner will be asked to move towards the angular corner.

\begin{theorem}\label{th4}
 All non-partitive configurations on an even $\times$ even square grid with all four corners occupied are gatherable despite at most one crash fault.
\end{theorem}

\begin{proof}
  As mentioned earlier, configurations having at least two robots other than the ones at the corners and at least one of them on the boundary are gatherable. So we consider only the cases with exactly 5 robots or that have no robots on the boundary except at the corners. We have the following cases to consider, based on the initial configuration:

  \textbf{Case 1:} Assume that the initial configuration has exactly five robots. Let $\lambda_{DA}$ be the largest sequence. Let $r_1$ be at $B$ and $r_2$ be the robot not at a corner. Now the algorithm asks $r_1, r_2$ to move towards $D$. If $r_1$ moves then we have a 3C2 configuration. Here too $r_1, r_2$ is to move towards $D$. If $r_1$ crashes and doesn't leave $B$, then $r_2$ should reach $D$. Once a  multiplicity is created at a corner, all non-faulty robots will move towards that corner.
  
  \textbf{Case 2:} Now assume that the configuration has at least 6 robots. First let the configuration be asymmetric or symmetric of first type with $D$ being the largest corner. Then let $r_1, r_2$ be the robots on $A$ and $C$ respectively. If both $r_1, r_2$ moves towards $D$, then the configuration has two corners occupied with $D$ being the largest corner. Then $r_1, r_2$ will continue to move towards $D$ and gathering will be accomplished according to theorem \ref{th2}. If only $r_2$ moves, then the configuration becomes 3C1. The pending move of $r_1$ is towards $D$, and is consistent with the algorithm for 3C1 configuration. Hence either the configuration becomes with two corners occupied with $D$ being the largest corner, or $r_2$ reaches $D$ first. In either case gathering will be achieved at $D$. The algorithm progresses exactly in the same manner if the configuration is symmetric of second type. $\square$

 \end{proof}
 
 \begin{figure}[thb!]
\centering

{
    \def\svgwidth{0.48\textwidth}
    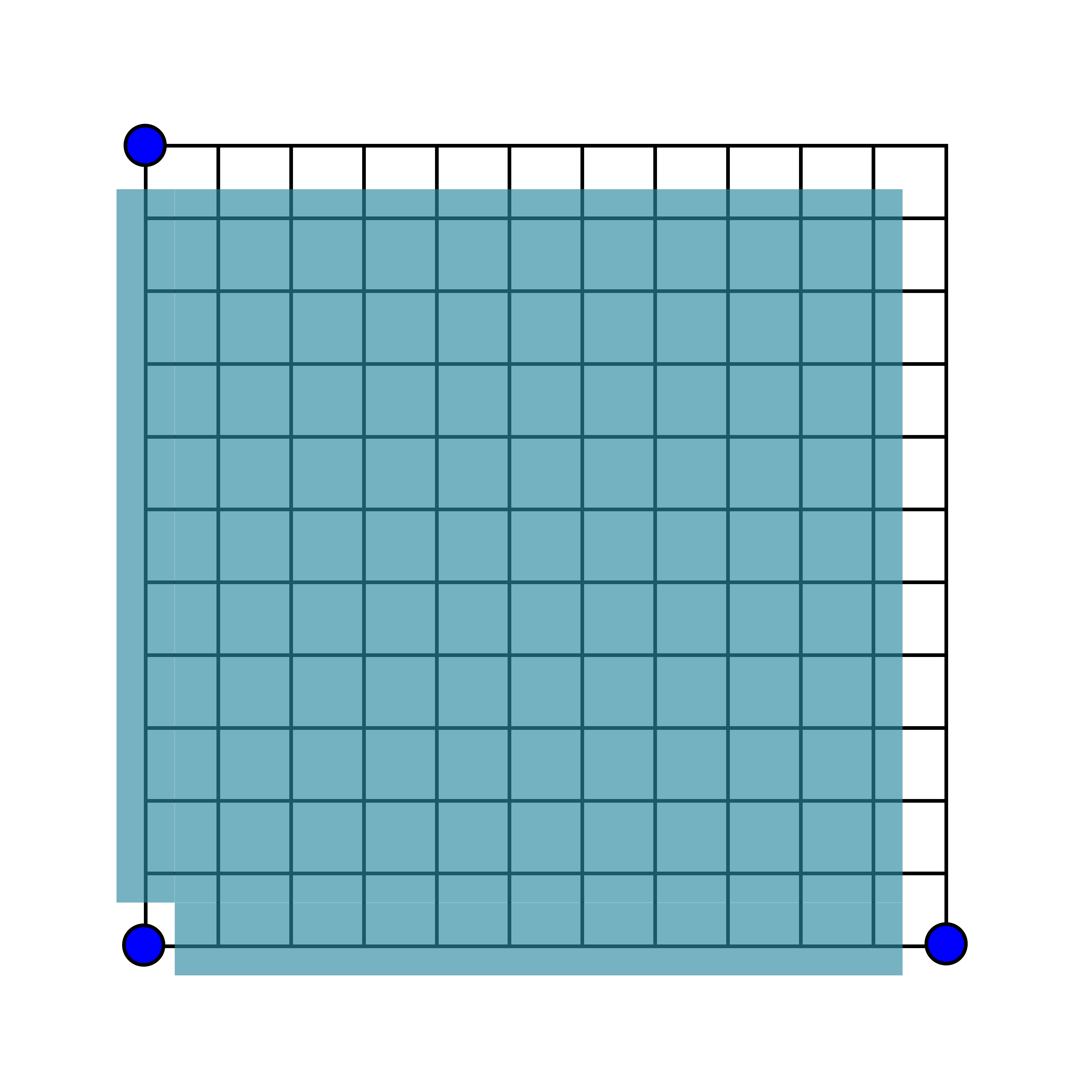
}

\caption[Short Caption]{}
\label{3region}
\end{figure}

Thus we have the following result.
 
 \begin{theorem}[Main result 1]
 
 If the initial configuration on a finite grid is non-partitive and not 2S2, then it is gatherable despite at most one crash fault.
 
\end{theorem}

\subsection{A unified gathering algorithm}\label{sec_uni}

 The algorithm we have presented in section \ref{sec4} gathers all non-partitive configurations with four corners occupied. However, the strategy we have used for configurations with three occupied corners as a subroutine of this algorithm, is different from the algorithm described in theorem \ref{th3}. If it is known beforehand that the initial configuration has exactly three corners occupied, then the algorithm described in theorem \ref{th3} can be used which works for all possible initial configurations. On the other hand if the initial configuration has four corners occupied, then the algorithm described in theorem \ref{th4} will work for all non-partitive initial configurations. However it is desired to have a unified gathering algorithm for all configurations. In order to achieve this we have to extend the strategy used for configurations 3C1 and 3C2 to other configurations with three corners occupied. However some specific configurations with exactly 6 or 7 robots have to be excluded. The configurations 3C1 and 3C2 are already defined in section \ref{sec4}. We classify the remaining configurations and describe the algorithm in each case in the following way. Again assume that the corners $D, A$ and $B$ are singly occupied.

\textbf{3C3 configuration:}\emph{3C3 configuration} includes the configuration with exactly 3 robots and configurations with exactly 4 robots but not 3C1. In these cases the two robots at $D$ and $B$ will be asked to move towards $A$, the angular corner.

\textbf{ 3C4 configuration:} Consider a configuration with at least 6 robots and which is not 3C1. If the shaded region shown in figure \ref{3region} contains at least two robots then it will be called a \emph{3C4 configuration}. In this all the robots in the shaded region will be asked to move towards the angular corner $A$.
 
 \textbf{3C5 configuration:} A configuration with at least 8 robots will be called a \emph{3C5 configuration} if it is not 3C1 or 3C4. In this case if the shaded region is empty, then the sides $DC$ and $CB$ have at least 5 robots other than the ones at the corner. If the shaded region has exactly one robot, then the sides $DC$ and $CB$ have at least 4 robots other than the ones at the corner. If one of these two sides, say $DC$, is empty, then 2 robots on $CB$ that are closest to $B$ will be asked to leave the boundary. If none of the sides are empty, then the robots on $DC$ and $CB$ that are closest to $D$ and $B$ respectively, will be asked to leave the boundary.

 
 \textbf{3C6 configuration:} A configuration with exactly 6 or 7 robots which is not 3C1 or 3C4 is called a \emph{3C6 configuration}.

\begin{theorem}[Main result 2]
 There is a unified gathering algorithm that gathers all non-partitive initial configurations on a finite grid except the 2S2 configuration and the 3C6 configuration.
\end{theorem}

\begin{proof}

For even $\times$ even sqaure grids except for configurations with exactly three occupied corners, the algorithm will be exactly the same as described in section \ref{sec_even}. From an initial configuration with four occupied corners, configurations 3C1 or 3C2 can be created. The algorithms for these configurations are already described in section \ref{sec4}, and it follows from the proof of theorem \ref{th4} that they can be incorporated the unified gathering algorithm. We only have to prove the same for the algorithms presented in this section.

 \textbf{Case 1:} Assume that the initial configuration is 3C3. Suppose that there are exactly 4 robots. Let $r_1, r_2, r_3$ be the robots at the corners $D, A$ and $B$ respectively. Let $r_4$ be the remaining robot. If both $r_1, r_3$ move simultaneously towards $A$, then the new configuration has exactly one corner occupied and we are done. So now assume that only one of $r_1, r_3$, say $r_1$, moves. After the move the new configuration has exactly two corners occupied. Since the initial configuration was not of 3C1, $r_4$ is not on $DC$ or $BC$. So $r_4$ is either at $DA,$ or $AB$ or somewhere in the interior of the grid. Note that regardless of where $r_4$ is, after the move by $r_1$, $A$ is the largest corner. So a possible pending move by $r_3$ is consistent with the algorithm for configurations with two corners occupied. Clearly eventually all the non-faulty robots will gather at $A$. The arguements are similar if the configuration has exactly 3 robots.
 
 \textbf{Case 2:} Assume that the initial configuration is 3C4. Clearly at least one of the robots in shaded region will reach $A$ in finite time. Note that the configuration remains 3C4 during the movements. Once a multiplicity is created at $A$ all the other non-faulty robots will go to $A$.
 
 \textbf{Case 3:} Now consider the 3C5 configuration. As the robots on the sides $DC$ and $CB$ leave the bounday, they enter the shaded region. If the shaded region initially contains exactly one robot, at most 2 new robots will enter the shaded region. If shaded region is initially empty, at most 3 new robots will enter the shaded region due to the asynchronous nature of the robots. Clearly after these moves the sides $DC$ and $CB$ together have at least 2 robots remaining on them (other than the ones at the corner). Hence the new configuration is not 3C1. Also after these moves the shaded region has now at least 2 robots. So the new configuration is 3C4. As shown earlier, gathering can be achieved from this configuration. $\square$

\end{proof}

\section{Conclusion}

We have shown that except one specific configuration called the 2S2 configuration, all configurations that are gatherable in a non-faulty system,  are also gatherable in presence of at most one crash fault. We have also devised a unified gathering algorithm for these initial configurations except for some marginal and specific configurations with exactly 6 or 7 robots on even $\times$ even square grids. The question left open is whether a fault-tolerant gathering algorithm can be given for the 2S2 configuration and be included in a unified algorithm. 

While there are some configurations where our algorithm asks many robots to move at the same time, in many cases at most two robots are allowed to move concurrently. A challenging direction of future research would be to further parallelize the movements of the robots so that they can survive multiple crash faults. This problem is valid for other graph topologies as well, as in most gathering algorithms in literature for weak robots on graphs, e.g. rings \cite{D14}, only a limited number of robots are allowed to move concurrently.


\bibliography{elsarticle-template.bib}

\end{document}